%% file: paper.tex
\documentclass[runningheads%
]{llncs}
\usepackage{geometry}
\usepackage[T1]{fontenc}
\usepackage[hidelinks]{hyperref}
\usepackage{cite}
\usepackage{url}
\usepackage{amsmath,amssymb}
\usepackage[]{todonotes}
\usepackage{tikzit}
\usetikzlibrary {arrows}
\usetikzlibrary{backgrounds}
\usepackage[ruled,linesnumbered,noend]{algorithm2e}
\input{complexes.tikzstyles}

\usepackage{hyperref}
\usepackage{comment}
\usepackage[capitalize]{cleveref}

\usepackage{graphicx}
\graphicspath{ {./Figures/} }

\DeclareMathOperator{\skel}{skel}
\DeclareMathOperator{\In}{In}
\DeclareMathOperator{\Ps}{Ps}
\DeclareMathOperator{\Ch}{Ch}
\DeclareMathOperator{\Bd}{Bd}
\DeclareMathOperator{\BdR}{BdR}
\DeclareMathOperator{\BdC}{BdC}

\DeclareMathOperator{\B}{\beta}

\newcommand{\G}{\mathcal{G}}
\newcommand\VI{\mathcal{V}_I}

\renewcommand{\P}{\mathcal{P}}
\newcommand{\N}{\mathcal{N}}
\newcommand{\RRG}{\mathcal{RRG}}
\newcommand{\C}{\mathcal{C}}
\renewcommand{\S}{\mathcal{S}}
\newcommand{\bD}{\mathbf{D}}
\newcommand{\bDo}{\bD^\omega}
\newcommand{\MA}{\bDo}
\newcommand{\otau}{\overline{\tau}}

\newcommand{\orr}{\overline{r}}

\newcommand{\oF}{\overline{F}}

\newcommand{\hF}{\hat{F}}
\newcommand{\hR}{\hat{R}}
\newcommand{\hB}{\hat{B}}

\newcommand{\K}{\mathcal{K}}

\begin{document}



\title{Topological Characterization of Consensus Solvability\\in Directed Dynamic Networks\thanks{This work is  supported by the Austrian Science Fund (FWF) projects ADynNet (P28182) and ByzDEL (P33600), and the German Research Foundation (DFG) project FlexNets (470029389). Kyrill Winkler was supported by the Vienna Science and Technology Fund (WWTF) project WHATIF (ICT19-045) and 
		Ami Paz by the Austrian Science Fund (FWF) and netIDEE SCIENCE project P 33775-N.
}}

\titlerunning{Topological Characterization of Consensus Solvability}

\author{Hugo Rincon Galeana\inst{1}
\and
Ulrich Schmid\inst{1}
\and
Kyrill Winkler\inst{2}
\and
Ami Paz\inst{3}
\and
Stefan Schmid\inst{4}
}
\authorrunning{H. Rincon, U. Schmid, K. Winkler, A. Paz, S. Schmid}
%
\institute{TU Wien,
Austria
\and
ITK Engineering, 
Austria
\and
CNRS, France
\and
TU Berlin, Germany
}
\maketitle              

\sloppy

\begin{abstract}
Consensus is one of the most fundamental problems in distributed computing. 
This paper studies the consensus problem in a synchronous dynamic directed network, in which communication is controlled by an oblivious message adversary. 
The question when consensus is possible in this model has already been studied thoroughly in the literature from a combinatorial perspective, and is known to be challenging. This paper presents a topological perspective on consensus solvability under oblivious message adversaries, which provides interesting new insights. 

Our main contribution is a topological characterization of consensus solvability,
which also leads to explicit decision procedures. 
Our approach is based on the novel notion of a communication pseudosphere, which can be seen as the message-passing analog of the well-known standard chromatic subdivision for wait-free shared memory systems. We further push the elegance and expressiveness of the ``geometric'' reasoning enabled by the topological approach by dealing with uninterpreted complexes, which considerably reduce the size of the protocol complex, and by labeling facets with information flow arrows, which give an intuitive meaning to the implicit epistemic status of the faces in a protocol complex.

\end{abstract}

\section{Introduction}

Consensus is a most fundamental problem in distributed computing, in which 
multiple processes need to agree on some value, based on their local inputs. 
The problem has already been studied for several decades and in various different models, yet in many distributed settings the question of when and how fast consensus can be achieved continues to puzzle researchers. 

This paper studies consensus in the fundamental setting where 
processes communicate over a synchronous dynamic directed network, where
communication is controlled by an oblivious \emph{message adversary}~\cite{AG13}.
This model is appealing, because it is conceptually simple and still provides a highly dynamic network model. 
In this model, fault-free processes communicate in a lock-step synchronous fashion using message passing, and a message adversary may drop some messages sent by the processes in each round.
Viewed more abstractly, the message adversary provides a sequence of directed communication graphs, whose edges indicate which process can successfully send a message to which other process in that round.
An oblivious message adversary is defined by a set $\bD$ of allowed
communication graphs, from which it can pick one in each round~\cite{CGP15}, independently of its picks in the previous rounds.

The model is practically motivated, as the communication topology of many large-scale distributed systems is \emph{dynamic} (e.g., due to mobility, interference, or failures) and its links are often \emph{asymmetric} (e.g., in optical or in wireless networks)~\cite{NKYG07}.
The model is also theoretically interesting, as solving consensus in general dynamic directed networks is known to be difficult~\cite{SW89,SWK09,CGP15,BRSSW18:TCS,WSS19:DC}.

Prior work primarily focused on the circumstances under which consensus is actually solvable under oblivious message adversaries~\cite{CGP15}. 
Only recently, 
first insights have been obtained on the time complexity of reaching consensus in this model~\cite{WPRSS23:ITCS}, using a combinatorial approach.
The present paper complements this by a topological perspective, which provides
interesting new insights and results.

\medskip
\noindent

\textbf{Our contributions:} 
Our main contribution is a topological characterization of consensus solvability 
for synchronous dynamic networks under oblivious message adversaries. It provides not only intuitive (``geometric'') explanations for the surprisingly
intricate time complexity results established in \cite{WPRSS23:ITCS},
both for the decision procedure (which allows to determine whether
consensus is solvable for a given oblivious message adversary or not)
and, in particular, for the termination time of any correct distributed
consensus algorithm.

To this end, we introduce the novel notion of a communication pseudosphere, which can be seen as the message-passing analog of the well-known standard chromatic subdivision for wait-free shared memory systems. Moreover, we use
uninterpreted complexes~\cite{SC20:PODC}, which considerably reduce the
size and structure of our protocol complexes. And last but not least,
following \cite{GP16:OPODIS}, we label the edges in our protocol complexes
by the information flow that they carry, which give a very intuitive meaning
to the the implicit epistemic status (regarding knowledge of initial values)
of the vertices/faces in a protocol complex. Together with the inherent beauty
and expressiveness of the topological approach, our tools facilitate
an almost ``geometric'' reasoning, which provides simple
and intuitive explanations for the surprising results of
\cite{WPRSS23:ITCS}, like the sometimes exponential gap between decision
complexity and consensus termination time. It also leads to
a novel decision procedure for deciding whether consensus under a given oblivious message adversary can be achieved
in some $k$ rounds.

In general, we believe that, unlike the combinatorial approaches
considered in the literature so far, our topological approach also has
the potential for the almost immediate generalization to other
decision problems and other message adversaries, and may hence be of
independent interest.


\medskip
\textbf{Related work}: 
Consensus problems arise in various models, including 
shared memory architectures, message-passing systems, and 
blockchains, among others~\cite{ongaro2014search,KO11:SIGACT,WS19:EATCS,abraham2017blockchain}.
The distributed consensus problem in the message-passing model, as it is considered in this paper, where communication occurs over a dynamic network, has been studied for almost 40 years~\cite{HRT02:ENTCS,SW07,SWK09,CBS09,BSW11:hyb,CGP15,CFN15:ICALP,FNS18:PODC}. Already in 1989, Santoro and Widmayer~\cite{SW89} showed that consensus is impossible in this model if 
up to $n-1$ messages may be lost each round. Schmid, Weiss and Keidar~\cite{SWK09} showed that if losses do not isolate the processes, consensus can even be solved when a quadratic number of messages is lost per round. Several other generalized models have been proposed in the literature~\cite{Gaf98,KS06,CBS09}, like the heard-of model by Charron-Bost and Schiper~\cite{CBS09}, and also different agreement problems like approximate and asymptotic consensus have been studied in these models~\cite{CFN15:ICALP,FNS18:PODC}. 
In many of these and similar works on consensus~\cite{FG11,BRS12:sirocco,SWS16:ICDCN,BRSSW18:TCS,WSS19:DC,NSW19:PODC,CastanedaFPRRT21topo}, a model is considered in which, in each round, a digraph is picked from a set of possible communication graphs. Afek and Gafni coined the term message adversary 
for this abstraction~\cite{AG13}, and used it for relating problems solvable in
wait-free read-write shared memory systems to those solvable in message-passing systems.
For a detailed overview of the field, we refer to the recent survey by Winkler and Schmid~\cite{WS19:EATCS}. 
%

An interesting alternative model for dynamic networks assumes a 
$T$-interval connectivity guarantee, that is, 
a common subgraph in the communication graphs
of every $T$ consecutive rounds~\cite{KLO10:STOC,KOM11}. 
In contrast to our directional model, solving consensus is relatively
simple here, since the $T$-interval connectivity model relies on 
bidirectional links and always connected communication graphs.
For example, $1$-interval-connectivity, the weakest form of
$T$-interval connectivity, implies that all nodes are able to reach
all the other nodes in the system in each of the graphs.
%
%
%
Solving consensus in undirected graphs that are always connected was also considered in the case of a given $(t+1)$-connected graph and at most $t$-node failures~\cite{CastanedaFPRRT19tres}.
Using graph theoretical tools, the authors extend the notion of a radius in a graph to determine the consensus termination time in the presence of failures.

Coulouma, Godard, and Peters~\cite{CGP15} showed an interesting equivalence relation, which captures the essence of consensus impossibility under oblivious message adversaries via the non-broadcastability of one of the so-called beta equivalence classes, hence refining the results of~\cite{SW07}. Building upon some of these insights, Winkler et al.~\cite{WPRSS23:ITCS}
studied of the time complexity of consensus in this model. In particular, they presented an explicit decision procedure and
analyzed both its decision time complexity and the termination time of distributed consensus. It not only turned out that 
consensus may take exponentially longer than broadcasting~\cite{itcs23broadcast}, but also that there is sometimes
an exponential gap between decision time and termination time. Surprisingly, this gap is not caused by properties related
to broadcastability of the beta classes, but rather by the number of those.

Whereas all the work discussed so far is combinatorial in nature, there is
also some related
topological research, see \cite{HKR13} for an introduction and overview.
Using topology in distributed computing started out from
wait-free computation in shared memory systems with immediate atomic snapshots (the IIS model), see e.g.~\cite{HS99:ACT,AttiyaC13,AttiyaCHP19,Kozlov15,Kozlov16,GRS22:ITCS}. The evolution of the protocol complex in the IIS model is governed
by the pivotal chromatic subdivision operation here. We will show that the latter
can alternatively be viewed as a specific oblivious message adversary, the set $\bD$ of which containins all transitively closed and unilaterally connected graphs.

Regarding topology in dynamic networks,
Casta{\~n}eda et al.~\cite{CastanedaFPRRT21topo} studied consensus and other problems in both static and dynamic graphs, albeit under the assumption that all the nodes know the graph sequence.
That is, they focused on the question of information dissemination, and put aside  questions of indistinguishability between graph sequences.
In contrast, in our paper, we develop a topological model that captures both information dissemination and indistinguishability. 
An adversarial model that falls into ``our'' class of models
has been considered by Godard and Perdereau~\cite{GP16:OPODIS}, who studied
general $k$-set agreement under the assumption that some maximum number of
(bidirectional) links could drop messages in a round. The authors also
introduced the idea to label edges in the protocol complex by arrows that
give the direction of the information flow, which we adopted.
Shimi and Casta{\~n}eda~\cite{SC20:PODC} studied $k$-set agreement
under the restricted class of oblivious message adversaries that
are ``closed-above'' (with $\bD$ containing, for every included graph,
also all graphs with more edges).

One of the challenges of applying topological tools in distributed settings is that the simplicial complex representing the system grows dramatically with the number of rounds, as well as with the number of processes and possible input values. In the case of colorless tasks, such as $k$-set agreement, the
attention can be restricted to colorless protocol complexes \cite{HKR13}. In the case
of the IIS model, its evolution is governed by the barycentric subdivision,
which results in much smaller protocol complexes than produced by the chromatic
subdivision. Unfortunately, however, it is not suitable for tracing indistinguishability
in dynamic networks under message adversaries. The same is true for the
``local protocol complexes'' introduced in~\cite{FraigniaudP20}. By contrast,
uninterpreted complexes, as introduced in \cite{SC20:PODC}, are effective
here and are hence also used in our paper.

Apart from consensus being a special case of $k$-set agreement (for $k=1$),
consensus has not been the primary problem of interest for topology in
distributed computing, in particular not for dynamic networks under
message adversaries. However,
a point-set topological characterization of when consensus is possible under general (both closed and non-closed) message adversaries has been presented by Nowak, Schmid and Winkler in \cite{NSW19:PODC}. The resulting decision procedure
is quite abstact, though (it acts on infinite admissible executions), 
and so are some results on the termination time for closed message adversaries that
confirm \cite{WSM19:OPODIS}.

The topology of message-passing models in general has been considered 
by Herlihy, Rajsbaum, and Tuttle already in 2002~\cite{HRT02:ENTCS}. 
Herlihy and Rajsbaum~\cite{HR10} studied $k$-set agreement in models
leading to shellable complexes.


\medskip
\noindent

\textbf{Paper organization:} 
We introduce our model of distributed computation and the oblivious message adversary in 
Section~\ref{sec:compmodel}.
In Section~\ref{sec:topologyintro} we present a framework which will allow us to study consensus on dynamic networks from a topological perspective.  
Our characterization of consensus solvability/impossibility for the oblivious
message adversary is presented in Section~\ref{sec:consensuschar}, where we also describe an explicit decision procedure.
In Section~\ref{sec:terminationtime} we further explore the relationship between 
the time complexity required by our decision procedure and
 the actual termination time of 
distributed consensus. 
We conclude our contribution and discuss future research directions in Section
\ref{sec:conclusion}.

\section{System Model}
\label{sec:compmodel}

We consider a synchronous dynamic network consisting of a set of $n$ processes that do not fail, which are fully-connected via point-to-point links that might drop messages.
We identify the processes solely by their unique ids, which are taken from the set 
$\Pi=\{p_1,\dots,p_n\}$ and known to the processes. Let $[n]=\{1,\dots,n\}$.
Processes execute a deterministic full-information protocol $P$, using broadcast (send-to-all) communication. Their execution proceeds in a sequence of lock-step rounds, where every process simultaneously broadcasts a message to every other process, without getting immediately informed of a successful message reception, and then computes its next state based on its current local 
state and the messages received in the round. The rounds are communication-closed, i.e., messages not received in a specific round are lost and will not be delivered later.

Communication is hence unreliable, and in fact controlled by an oblivious \emph{message adversary} (MA) with non-empty graph set $\bD=\{D_1,\dots,D_k\}$.
All the graphs have $\Pi$ as their set of nodes, and an edge $p_i\to p_j$ represents a communication link from $p_i$ to $p_j$.
For every round $r\geq 1$, the MA arbitrarily picks
some communication graph $G_r$ from $\bD$,
and a message from a process $p_i$ arrives to process $p_j$ in this round if $G_r$ contains the edge $p_i\to p_j$, and otherwise it is lost.
We assume processes have persistent memory, i.e., every graph in $\bD$ contains all self-loops $p_i\to p_i$.
An infinite graph sequence $\G=(G_r)_{r\ge 1}$ picked by the message adversary is called a feasible graph sequence, 
and $\bD^\omega$ denotes the set of all \emph{feasible graph sequences} for the oblivious message adversary with
graph set $\bD$. 
The processes know $\bD$, but they do not have a priori knowledge of the graph $G_r$ for any $r$ (though they may infer it after the round occurred).

We consider a system where the global state is fully determined by the local states of each process. Therefore, 
a \emph{configuration} is just the vector of the local states (also called \emph{views}) of the processes. 
An admissible \emph{execution} $a$ of $P$ is just the sequence 
of configurations $a = (a_r)_{r \ge 0}$ at the end of the rounds $r\geq 1$, induced 
by a feasible graph sequence $\G\in \bDo$ starting out from a given initial configuration $a_0$.
Since we will restrict our attention to deterministic protocols $P$, the graph sequence $\G$ and 
the initial configuration $a_0$ uniquely determine~$a$. The \emph{view} of process 
$p_i$ in $a_r$ at the end of round $r\geq 1$ is denoted as $a_r(p_i)$; its initial view is 
denoted as $a_0(p_i)$.

We restrict our attention to deterministic protocols for the consensus problem, defined as follows:

\begin{definition}[Consensus]\label{def:consensus}
Every process $p_i\in \Pi$ has an input value $x_i\in \VI$ taken from 
a finite input domain $\VI$, which is encoded in the initial state, 
and an output value $y_i\in\VI \cup \{\bot\}$, initially $y_i=\bot$. 
In every admissible execution, a correct consensus protocol $P$ must ensure
the following properties: 
\begin{itemize}
\item \textbf{Termination:} Eventually,
every $p_i \in \Pi$ must decide, i.e., change to $y_i\neq\bot$, exactly once. 
\item \textbf{Agreement:} If processes $p_i$ and $p_j$ have decided, then $y_i=y_j$. 
\item \textbf{(Strong) Validity:} If $y_i\neq \bot$, then $y_i=x_j$ for some $p_j\in\Pi$, i.e., must be
the input value of some process $p_j$.
\end{itemize}
\end{definition}

In any given admissible execution $a$ of $P$, induced by $\G\in\MA$, for a process $p_i$, let $\In^{\G}(p_i,r)$ be the set 
of processes $p_i$ has \emph{heard of} in round $r$ (see also~\cite{CBS09}), i.e., 
the set of in-neighbors of process $p_i$ in $G_r$, and $\In^{\G}(p_i,0)=\{p_i\}$. 
Since all graphs in $\bD$ contain all self-loops, we have that $p_i \in \In^{\G}(p_i, r)$ for all $r\geq 0$ and $p_i\in\Pi$.
If the round $r$ is clear from the context, we also abbreviate $\In^{\G}(p_i)=\In^{\G}(p_i, r)$.

The evolution of the local views of the processes in an admissible execution $a$, induced by $\G$ and
the initial configuration $a_0$, can now be defined recursively as
\begin{equation}
a_r(p_i) = \bigl\{ (p_j, r, a_{r-1}(p_j)) : p_j \in \In^{\G}(p_i,r) \bigr\} \quad\mbox{for $r > 0$}.\label{eq:viewsa}
\end{equation}
Note that we could drop the round number $r$ from $(p_j, r, a_{r-1}(p_j))$ in the above
definition, since it is implicitly contained in the structure of $a_r(p_i)$; we included
it explicitly for clarity only.
The set of all possible round-$r$ views of $p_i$, including all the initial views for $r=0$, 
in any admissible execution, is denoted by 
$A^r(p_i)=\{a_r(p_i) \mid \forall \mbox{ admissible executions $a$ under MA}\}$.

In any admissible execution $a$, every process must eventually reach a final view, where it can take a decision on an output value which will not be changed later.
Consequently, there is some final round after which all processes have decided.


\section{A Topological Framework for Consensus}
\label{sec:topologyintro}
 
In this section, we introduce the basic elements of combinatorial topology and
specific concepts needed in our context of synchronous message-passing networks. 

Combinatorial topology in distributed computing~\cite{HKR13} rests on
simplicial input and output complexes describing the feasible input and output values
of a distributed decision task like consensus, and a carrier map that defines
the allowed output value(s), i.e., output simplices, for a given input simplex. 
A protocol that solves such a task in some computational model
gives rise to another simplicial complex, the protocol complex, which describes the evolution of the local views of the processes in any execution.
Protocol complexes traditionally model full information protocols in round-based models, 
which ensures a well-organized structure: 
The processes execute a sequence of communication operations, which disseminate their 
complete views, until they are able to make a decision.
Finally, a protocol induces a simplicial decision map, 
which maps each vertex in the protocol complex to an output vertex in a way compatible with the carrier map.

\subsection{Basic topological definitions}
\label{sec:standardtopology}

We start with the definitions of the basic vocabulary of combinatorial topology:

\begin{definition}[Abstract simplicial complex]\label{def:abssimpcomplex}
	An \emph{abstract simplicial complex} $\mathcal K$ is a pair $\langle V(\mathcal{K}), F(\mathcal{K}) \rangle$, where $V(\mathcal{K})$ is a set, $F(\mathcal{K}) \subseteq 2^{V(\mathcal{K})}$, and for any $\sigma, \tau \in 2^{V(\mathcal{K})}$ such that $\sigma \subseteq \tau$ and $\tau \in F(\mathcal{K})$, then $\sigma \in F(\mathcal{K})$. 
	$V(\mathcal{K})$ is called the set of \emph{vertices}, and $F(\mathcal{K})$ is the set of \emph{faces} or \emph{simplices} 
    of $\mathcal{K}$. 
    We say that a simplex $\sigma$ is a \emph{facet} if it is maximal with respect to containment, and a \emph{proper face} otherwise. 
    We use $Fct(\K)$ to denote the set of all facets of $\K$, and note that for a given $V(\K)$ we have that $F(\K)$ uniquely define  $Fct(\K)$ and vice versa. 
    A simplicial complex is
    \emph{finite} if its vertex set is finite, which will be the case for all the complexes in this paper.
\end{definition}
All the simplicial complexes we consider in this work are abstract.
For conciseness, we will usually sloppily write $\sigma \in \K$ instead of $\sigma \in F(\K)$.

\begin{definition}[Subcomplex]
	Let $\mathcal{K}$ and $\mathcal{L}$ be  simplicial complexes. We say that $\mathcal{L}$ is a \emph{subcomplex} of $\mathcal{K}$, written as $\mathcal{L} \subseteq \mathcal{K}$, if $V(\mathcal{L}) \subseteq V(\mathcal{K})$ and $F(\mathcal{L}) \subseteq F(\mathcal{K})$.
\end{definition}

\begin{definition}[Dimension] \label{def:dimension}
	Let $\mathcal{K}$ be a simplicial complex, and $\sigma \in F(\mathcal{K})$ be a simplex. We say that $\sigma$ has \emph{dimension} $k$, denoted by $\dim(\sigma)=k$, if it has a cardinality of $k+1$. A simplicial complex $\K$ is of dimension $k$ if every facet has dimension at most $k$, and it is \emph{pure} if all its facets have the same dimension.
\end{definition}
We sometimes denote a simplex as $\sigma^k$ in order to stress that its dimension is $k$.

\begin{definition}[Skeletons and boundary complex]
The \emph{$k$-skeleton} $\skel^k(\K)$ of a simplicial complex $\K$ is the subcomplex consisting of all simplices of
dimension at most $k$.
  The \emph{boundary complex} $\partial \sigma$ of a simplex $\sigma$, viewed as a complex, is the complex made up of all proper faces of $\sigma$.
\end{definition}

\begin{definition}[Simplicial maps] \label{def:simpmap}
	Let $\mathcal{K}$ and $\mathcal{L}$ be simplicial complexes. We say that a vertex map $\mu : V(\mathcal{K}) \rightarrow V(\mathcal{L})$ is a \emph{simplicial map} if, for any $\sigma \in F(\mathcal{K})$, $\mu (\sigma) \in F(\mathcal{L})$;
	here, $\mu (\sigma)=\{\mu(v)\mid v\in\sigma\}$.	
\end{definition}

\begin{definition}[Colorings and chromatic simplicial complexes]\label{def:coloring}
  We say that a simplicial complex $\K$ has a \emph{proper $c$-coloring} $\chi$, if there exists $\chi:V(\K) \to \{p_1,p_2, \ldots, p_{c}\}$ that is injective at every face of $\K$.
  If $\mathcal{K}$ has a proper $(\dim(\K)+1)$-coloring, we say it is a \emph{chromatic} simplicial complex.
\end{definition}

The range of $\chi$ is extended to sets of vertices $S$ by defining $\chi(S)=\{\chi(v)\mid v\in S\}$, which implies e.g.~$\chi(\sigma)=\chi(V(\sigma))$.

\begin{definition}[Carrier Map]\label{def:carrmap}
	Let $\mathcal{K}$ and $\mathcal{L}$ be simplicial complexes and $\Phi: F(\mathcal{K}) \rightarrow 2^{\mathcal{L}}$. We say that $\Phi$ is a \emph{carrier map}, if $\Phi(\sigma)$ is a subcomplex of $\mathcal{L}$ for any $\sigma \in \mathcal{K}$, and for any $\sigma_1, \sigma_2 \in \mathcal{L}$, $\Phi( \sigma_1 \cap \sigma_2) \subseteq \Phi(\sigma_1) \cap \Phi(\sigma_2)$.
	
	We say that a carrier map is \emph{rigid} if it maps every simplex $\sigma \in \mathcal{K}$ to a complex $\Phi(\sigma)$ which is pure of dimension $\dim(\sigma)$.
	It is said to be \emph{strict} if that for any two simplices $\sigma, \tau \in \mathcal{K}$, $\Phi(\sigma \cap \tau) = \Phi(\sigma) \cap \Phi(\tau)$.

	We say that a carrier map $\Phi: \mathcal{K} \rightarrow 2^\mathcal{L}$ \emph{carries} a simplicial vertex map $\mu: V(\mathcal{K}) \rightarrow V(\mathcal{L}$) if for any $\sigma \in \mathcal{K}$, $\mu(\sigma) \in \Phi (\sigma)$.
\end{definition}

Having introduced our basic vocabulary, we can now define the main ingredients
for the topological modeling of consensus in our setting. 

Generally, a \emph{distributed task} is defined by a tuple $T = \langle \mathcal{I}, \mathcal{O}, \Delta \rangle$ consisting of chromatic simplicial complexes $\mathcal{I}$ and $\mathcal{O}$  that model the valid input and output configurations respectively, for the set $\Pi$ of processes, and $\Delta : \mathcal{I} \rightarrow 2^{\mathcal{O}}$ is a carrier map that maps valid input configurations to sets of valid output configurations. 
Both complexes have vertices of the form $(p_i,x)$ with $p_i\in \Pi$, and they are chromatic with the coloring function $\chi((p_i,x))=p_i$.
All the simplicial maps we consider in this work are \emph{color preserving}, in the sense that they map each vertex $(p_i,x)$ to a vertex $(p_i,y)$ with the same process id $p_i$.

Many interesting tasks have some degree of regularity (that is, symmetry) in the input complex. 
In the case of consensus, in particular, any combination
of input values from $\VI$ is a legitimate initial configuration. Consequently,
the input complex for consensus in the classic topological modeling is
a pseudosphere~\cite{HRT02:ENTCS}.

In this paper, we will exploit the fact that strong validity does not force us to
individually trace the evolution of every possible initial configuration of the protocol complex. 
We will therefore restrict our attention to 
\emph{uninterpreted complexes}~\cite{SC20:PODC}: Instead of providing different
vertices for every possible value of $x_i$, we provide only one vertex labeled
with $\{p_i\}$, carrying the meaning of ``the actual input value $x_i$ of $p_i$''. This way, we can abstract away the input domain $\VI$ as well as the
actual assignment of initial values $x_i\in\VI$ to the processes. 
Topologically, uninterpreted complexes thus correspond to a ``flattening'' of the standard
complexes with respect to all input and output values. The main advantages of
resorting to uninterpreted protocol complexes is that they are exponentially
smaller than the standard protocol complex, even in the case of binary consensus, and independent of the particular initial configuration.
This can be compared with the study of \emph{colorless tasks}~\cite[Ch. 4]{HKR13}, where a different form of ``flattening'' of the complexes is done by omitting the process ids.

\begin{definition}[Uninterpreted input complex for consensus]\label{def:uicomplex}
The \emph{uninterpreted input complex} $\mathcal{I}$ for consensus is just a 
single \emph{initial simplex} $\sigma_0=\{(p_1,\{p_1\}), \dots, (p_n,\{p_n\})\}$ 
and all its faces, with the set of vertices 
$V(\mathcal{I}) = V(\sigma_0) = \{ (p_i,\{p_i\}) \; \vert \; p_i \in \Pi \}$, 
where the label $\{p_i\}$ represents the ``uninterpreted'' (i.e., fixed but arbitrary) 
input value of $p_i$.
\end{definition}
We use $\sigma_0$ throughout this paper to denote the above input simplex.

The \emph{uninterpreted output complex} $\mathcal{O}$ for consensus just specifies the process whose
input value will determine the decision value.

\begin{definition}[Uninterpreted output complex for consensus]\label{def:uocomplex}
The \emph{uninterpreted output complex} $\mathcal{O}$ for consensus is the 
union of $n$ disjoint complexes $\mathcal{O}(p_j)$, $p_j\in\Pi$, each
consisting of the simplex 
$\{(p_1,\{p_j\}), \dots, (p_n,\{p_j\})\}$ and all its faces. The label $\{p_j\}$ 
represents the ``uninterpreted'' (i.e., fixed but arbitrary) input value of $p_j$.
\end{definition}


The carrier map $\Delta$
for the consensus task maps any face $\rho$ of the initial simplex $\sigma_0 \in \mathcal{I}$ 
to $\dim(\rho)$-faces of $\mathcal{O}$ that all have a coloring equal to $\chi(\rho)$. 
Clearly, $\Delta$ is rigid and strict.

The \emph{uninterpreted protocol complex} $\P_r^{\MA}$ consists of vertices that are labeled by the heard-of histories the corresponding process has been able
to gather so far. 

\begin{definition}[Heard-of histories]\label{def:HOhistory}
For a feasible graph sequence $\G$, the \emph{heard-of history} $h_r^{\G}(p_i)$ of a process $p_i$ at the end of round~$r$ is defined as
\begin{eqnarray}
h_r^{\G}(p_i)&=&\{(p_j,h_{r-1}^{\G}(p_j)) \mid p_j \in \In^{\G}(p_i,r)\} \mbox{ for $r\geq 1$}, \label{eq:hr}\\
h_0^{\G}(p_i)&=& \{p_i\}.
\end{eqnarray}
The global heard-of history $h_r^{\G}$ at the end of round $r$ is just the tuple $(h_r^{\G}(p_1),\dots,h_r^{\G}(p_n))$.

The set of processes $p_i$ has ever heard of up to $h_r^{\G}(p_i)$, i.e., the end of round $r$, is denoted
$\cup h_r^{\G}(p_i) = \bigcup_{p_j \in \In^{\G}(p_i,r)} \cup h_{r-1}^{\G}(p_j)$ and $\cup h_0^{\G}(p_i)=h_0^{\G}(p_i)=\{p_i\}$. 

The set of all possible heard-of histories of $p_i$ (resp.\ the global ones) at the end of round $r\geq 0$, in every feasible graph sequence $\G\in\MA$, is denoted by 
\begin{eqnarray}
H^r(p_i)&=&\{h_r^{\G}(p_i) \mid  \G\in\MA \},\\
H^r&=&\{(h_r^{\G}(p_1),\dots,h_r^{\G}(p_n)) \mid  \G\in\MA\}.
\end{eqnarray}
\end{definition}

The uninterpreted protocol complex $\P_r^{\MA}$, which does not depend on the
initial configuration but only on $\MA$, is defined as follows: 

\begin{definition}[Uninterpreted protocol complex for $\MA$]\label{def:upcomplex}
The \emph{uninterpreted $r$-round protocol complex} $\P_r^{\MA} = \langle V(\P_r^{\MA}) , F(\P_r^{\MA}) \rangle$, $r\geq 0$, for a given oblivious message adversary $\MA$, is defined by its vertices and facets as follows:
\begin{eqnarray}
V(\P_r^{\MA}) &=& \bigl\{(p_i,h_r(p_i)) \; \vert \; p_i \in \Pi, \; h_r(p_i) \in H^r(p_i) \bigr\}, \nonumber\\
Fct(\P_r^{\MA}) &=& 
\bigl\{ \{(p_1,h_r(p_1)),\dots, (p_n,h_r(p_n))\} | \forall 1 \leq i \leq n: p_i \in \Pi, (h_r(p_1),\dots,h_n(p_n)) \in H^r\bigr\}\nonumber.
\end{eqnarray}
\end{definition}
For conciseness, we will often omit the superscript $\MA$ when the 
oblivious message adversary considered is clear from the context.

The \emph{decision map} $\mu: V(\P_r^{\MA}) \rightarrow V(\mathcal{O})$ is a chromatic simplicial map that maps a final view of a process $p_i$ at the end of round $r$ to an output value $p_j$ such that $p_j\in \cup h_r^{\G}(p_i)$; 
it is not defined for non-final views. 
Note that $\mu$ is uniquely determined by the images of the facets in $\P_r^{\MA}$
after any round $r$ were all processes have final views.
We say that consensus is solvable if such a simplicial map $\mu$ exists.

\paragraph{Remark.}
Standard topological modeling, which does not utilize uninterpreted complexes, also requires an
execution carrier map $\Xi: \mathcal{I} \to 2^{\P}$, which defines the subcomplex
$\Xi(\sigma)$ of the protocol complex $\P$ that arises when the protocol starts from 
the initial simplex $\sigma \in \mathcal{I}$. Solving a task requires $\mu \circ \Xi$ 
to be carried by $\Delta$, i.e., $\mu(\Xi(\sigma)) \in \Delta(\sigma)$ for all 
$\sigma \in \mathcal{I}$. In our setting, since we have only one (uninterpreted) facet 
in our input complex $\sigma_0$ and a protocol complex that can be written as 
$\bigcup_{r\geq 1} \P_r^{\MA}=\bigcup_{r\geq 1} \P^r(\sigma_0)$ (i.e., the union 
of all iterated protocol complex construction operators $\P^r$ given in
\cref{def:Pps} below), both the execution carrier map $\Xi$ and the carrier map $\Delta$ are independent of 
the actual initial values and hence quite simple: The former is just $\Xi=\bigcup_{r\geq 1} \P^r$
(with every $\P^r$ viewed as a carrier map), the latter has been stated after \cref{def:uocomplex}.

\subsection{Communication pseudospheres}

Rather than directly using \cref{def:upcomplex} for $\P_r$, we will now introduce an 
alternative definition based on communication pseudospheres. The latter
can be seen as the the message-passing analogon of the
well-known standard chromatic subdivision (see \cref{def:chrsub}) for wait-free shared memory systems. Topologically, it can be defined as follows:

\begin{definition}[Communication pseudosphere]\label{def:recps}
  Let $\mathcal{K}$ be an $(n-1)$-dimensional pure  simplicial complex
  with a proper coloring $\chi: V(\mathcal{K}) \rightarrow \{p_1, \ldots, p_n\}$. We define the \emph{communication pseudosphere} $\Ps(\K)$ through its vertex set and facets as follows:
\begin{align}		
	V(\Ps(\mathcal{K})) &= \bigl\{(p_i,\sigma) \; | \; \sigma \in F(\mathcal{K}), p_i \in \chi(\sigma)\bigr\}, \label{eq:vertexPs}\\
	Fct(\Ps(\mathcal{K})) &= \bigl\{\{(p_{1}, \sigma_1), (p_{2}, \sigma_2), \ldots (p_{n}, \sigma_n)\} \; | \;  \forall 1 \leq i \leq n: \; \sigma_i \in F(\K), p_i \in \chi(\sigma_i) 
        \bigr\}.\label{eq:facePs} 
\end{align}	
\end{definition}

Given an $(n-1)$-dimensional simplex $\sigma^{n-1} = \{(p_1,h_1), \dots, (p_n,h_n)\} \in \K$, the
communication pseudosphere $\Ps(\sigma^{n-1})$ contains a vertex $(p_i,\sigma)$
for every subset $\sigma\subseteq\{(p_1,h_1), \dots, (p_n,h_n)\}$ that satisfies $\{(p_i,h_i)\}\in \sigma$. 
Intuitively, $\sigma$ represents the information of those processes $p_i$ could have heard of in 
a round (recall that $p_i$ always hears of itself). $\Ps(\sigma^{n-1})$ hence indeed matches the 
definition of a pseudosphere~\cite{HRT02:ENTCS}.

Since $\bigl|\bigl\{
\sigma\subseteq\{(p_1,h_1), \dots, (p_n,h_n)\}\setminus\{(p_i,h_i)\}\bigr\}\bigr|=2^{n-1}$ 
for every $p_i$, every communication pseudosphere $\Ps(\sigma^{n-1})$
consists of $|V(\Ps(\sigma^{n-1}))|=n2^{n-1}$ vertices:
For every given vertex $(p_i,\sigma)$ and every $p_j\neq p_i$, there are exactly $2^{n-1}$ differently 
labeled vertices $(p_j,\cdot)$. Since $(p_i,\sigma)$ has an edge to each of those in the complex 
$\Ps(\sigma^{n-1})$, its degree must hence be $d=(n-1)2^{n-1}$. 

In the case of $n=2$ or $n=3$, 
let $v=|V(\Ps(\sigma^{n-1}))|$,
$e=|E(\Ps(\sigma^{n-1}))|$ and 
$f=|Fct(\Ps(\sigma^{n-1}))|$
denote the numbers of vertices, edges and facets in $\Ps(\sigma^{n-1})$, respectively.
It obviously holds that $v\cdot d=2e$ and $v\cdot d=nf$. 
Therefore, 
$e=vd/2=n(n-1)2^{2(n-1)-1}$ and $f=vd/n=(n-1)2^{2(n-1)}$. 
For $n=2$, we thus get $v=4$, $f=e=4$, $d=2$ and hence the following communication
pseudosphere $\Ps(\sigma_0^1)$ for the initial simplex $\sigma_0^1=\{(p_r,\{p_r\}), (p_w,\{p_w\})\}$:
\begin{equation}
  \input{Figures/n2all.tikz}\label{eq:graph}
\end{equation}

In the above figure, and throughout this paper, we use the labeling convention of the edges
proposed in~\cite{GP16:OPODIS}, which indicates the information flow
between the vertices in a simplex. For example, in the middle simplex (connected with edge $\leftrightarrow$), 
both processes have heard
from each other in round 1, so the connecting edge is denoted by 
$\leftrightarrow$. An edge without any arrow means that the two
endpoints do not hear from each other. Note carefully that we will
incorporporate these arrows also when talking about facets and faces 
that are \emph{isomorphic}: Throughout this paper, two faces $\sigma$
and $\kappa$ arising in our protocol complexes will be considered isomorphic 
only if $\chi(\sigma)=\chi(\kappa)$ and if all edges have the same orientation.

We note also that the labeling
of the vertices with the faces of $\sigma_0^1$ is highly
redundant. We will hence condense vertex labels when
we need to refer to them explicitly, and e.g.\ write
$(p_r,\{p_r,p_w\})$ instead of $(p_r,\{(p_r,\{p_r\}),(p_w,\{p_w\})\})$.

The communication pseudosphere  $\Ps(\sigma_0^2)$ for the initial simplex $\sigma_0^2=\{(p_r,\{p_r\}),(p_g,\{p_g\}), (p_w,\{p_w\})\}$ for $n=3$ is depicted in \cref{fig:pseudospheres}. It also highlights two facets, corresponding to the graphs 
$G_1$ (grey) and $G_2$ (yellow):
\begin{equation}\label{eq:advcomm}
  \input{Figures/graphs.tikz}
\end{equation}

\begin{figure}
\ctikzfig{Figures/n3all_bent_withpolygon}
\caption{Communication pseudosphere $\Ps(\sigma^{n-1})$ for $n=3$ (where $L=4$, $V=12$, $E=48$, $F=32$ and $d=8$), 
with the communication graphs of \cref{eq:advcomm} highlighted. Thick
edges represent the standard chromatic subdivision $\Ch(\sigma^2)$.}
\label{fig:pseudospheres}
\end{figure}

We will now recast
the definition of the uninterpreted protocol complex for a given oblivious
message adversary $\MA$ in terms of a communication pseudosphere. Recall from \cref{def:upcomplex}
that the uninterpreted initial protocol complex $\P_0=\P_0^{\MA}$
only consists of the single initial simplex
$\sigma_0 = \sigma_0^{n-1} = \{(p_1,\{p_1\}), \dots, (p_n,\{p_n\})\}$ and all
its faces. It represents the uninterpreted initial state, where every process 
has heard only from itself. Here is an example for $n=3$ and $\Pi=\{p_w,p_r,p_g\}$:
\begin{equation}
  \input{Figures/P0.tikz}
\end{equation}

Consequently, the single-round protocol complex $\P_1=\P_1^{\MA}$ is just the subcomplex of the 
communication pseudosphere $\Ps(\sigma_0)$ induced by the set $\bD$ of possible graphs. 
For example, $\P_1$ for $\bD=\{G_1,G_2\}$ is the subcomplex
of $\Ps(\sigma_0)$ made up by the two highlighted
facets corresponding to the graphs $G_1$ and $G_2$ in \cref{fig:pseudospheres}. 
That is, rather than labeling the vertices of $\Ps(\sigma_0)$
with \emph{all} the possible subsets of faces of $\sigma_0$ as in \cref{def:recps}, 
only those faces that are communicated via one of the graphs 
in $\bD$ are used by the 
\emph{protocol complex construction operator} $\P=\P^{\MA}$ for generating
$\P_1=\P(\sigma_0)$. Conversely, if one interprets $\bD$ as an $(n-1)$-dimensional simplicial 
complex $\bD(\sigma_0)$, consisting of one facet (and all its
faces) per graph $G\in \bD$ according to (\ref{eq:facePs}), one could write
$\P^{\MA}(\sigma_0)=\Ps(\sigma_0) \cap \bD(\sigma_0)$. 

This can be compactly summarized in the following definition: 

\begin{definition}[Protocol complex construction pseudosphere]\label{def:Pps}
  Let $\mathcal{K}$ be an $(n-1)$-dimensional pure simplicial complex
  with a proper coloring $\chi: V(\mathcal{K}) \rightarrow \{p_1, \ldots, p_n\}$,
  and $\In^G(p_i)$ be the set of processes that $p_i$ hears of in the communication graph $G \in \bD$.
  We define the \emph{protocol complex construction pseudosphere} $\P(\K)$ for
  the message adversary $\MA$, induced by the operator $\P:Fct(\K) \to \P(\K)$
  that can be applied to the facets of $\K$, through its vertex set and facets as follows:
\begin{align}		
	V(\P(\mathcal{K})) &= \bigl\{(p_i,\sigma) \in \Pi\times F(\K)\; | \; \exists G \in \bD: \In^G(p_i)=\chi(\sigma)\bigr\}, \label{eq:vertexN}\\
	Fct(\P(\mathcal{K})) &= \bigl\{\{(p_{1}, \sigma_1), \ldots, (p_{n}, \sigma_n)\} \; | \;  \exists G \in \bD, \forall 1 \leq i \leq n: 
	\In^G(p_i)=\chi(\sigma_i) 
        \bigr\}.\label{eq:faceN}
\end{align}	
\end{definition}

According to \cref{def:Pps}, 
our operator $\P$ (as well as $\Ps$) is actually defined only for the
facets in $\K$,
i.e., the dimension $n-1$ is actually implicitly encoded in the operator.
We will establish below that this is sufficient for our purposes, since 
every $\P$ is \emph{boundary consistent}: This property will allow us 
to uniquely define $\P$ for proper faces in $\K$ as well. 
We will use the following simple definition of boundary consistency, which 
makes use of the fact that the proper coloring of the vertices of a chromatic simplicial 
complex defines a natural ordering of the vertices of any of its faces.

\begin{definition}[Boundary consistency]
	\label{def:bc}
	We say that a protocol complex construction operator $\P$ according to \cref{def:Pps}
	is \emph{boundary consistent}, if for all
	possible choices of three facets $\sigma$, $\kappa$ and $\tau$ from every simplicial complex on which $\P$ can be applied, it holds that 
	\begin{equation}
		\sigma \cap \kappa = \sigma \cap \tau
		\implies
		\P(\sigma)\cap\P(\kappa)=
		\P(\sigma)\cap\P(\tau). \label{eq:bc}
	\end{equation}
\end{definition}

The following \cref{lem:bcP} shows that every $\P$ is boundary consistent
and that one can uniquely define $\P(\rho)$ also for a non-maximal simplex 
$\rho$ (taken as a complex). Moreover, it reveals that $\P$, viewed as
a carrier map, is strict (but not necessarily rigid):

\begin{lemma}[Boundary consistency of $\P$]\label{lem:bcP}
Every protocol construction operator $\P$ according to \cref{def:Pps} is boundary
consistent. It can be applied to any simplex $\rho\in\K$, viewed as a complex, and 
produces a unique (possibly impure) chromatic complex $\P(\rho)$ with dimension at 
most $\dim(\rho)$. Moreover,
\begin{equation}
\P(\sigma \cap \kappa) = \P(\sigma) \cap \P(\kappa)\label{eq:Pstrict}
\end{equation}
for any $\sigma, \kappa \in \K$.
\end{lemma}
\begin{proof}
Using the notation from \cref{def:bc}, assume $\rho=\sigma \cap \kappa = \sigma \cap \tau$
for $0 \leq \dim(\rho) < n-1$; for the remaining cases, \cref{eq:Pstrict} holds
trivially. Consider the facet
$F_\sigma=\{(p_{1}, \sigma_1), (p_{2}, \sigma_2), \ldots (p_{n}, \sigma_n)\}$ 
resp.\
$F_\kappa=\{(p_{1}, \kappa_1), (p_{2}, \kappa_2), \ldots (p_{n}, \kappa_n)\}$ caused
by the same graph $G\in\bD$ in $\P(\sigma)$ resp.\ $\P(\kappa)$ according to \cref{eq:faceN}.
Recall that $\sigma_i$ resp.\ $\kappa_i$ is a face of $\sigma$ resp.\ $\kappa$ that represents
the information $p_i$ receives from the processes in $\chi(\sigma_i)=\chi(\kappa_i)$ via 
$\In^G(p_i)$. 

A vertex $(p_{i}, \kappa_i)$ appears in $\P(\sigma)\cap \P(\kappa)$ if and only if
$\kappa_i=\sigma_i$, which, in turn, holds only if $\chi(\kappa_i) \subseteq \chi(\rho)$.
Indeed, if $\kappa_i$ would contain just one vertex $v\in V(\kappa)$ with $\chi(v) \in 
\chi(\kappa\setminus\rho)$, then $\sigma_i$ would contain the corresponding vertex 
$v'\in V(\sigma)$ with $\chi(v')=\chi(v)$ satisfying $v'\neq v$ since
$(\kappa\setminus\rho) \cap (\sigma\setminus\rho) = \emptyset$, by the definition of $\rho$.
This would contradict $\kappa_i=\sigma_i$, however. Note that, since $p_i\in \chi(\sigma_i)$ 
for every $i$, this also implies $p_i \in \chi(\rho)$. 

Consequently, it is precisely the maximal face in $F_\sigma$ (and in $F_\kappa$) 
consisting only of identical vertices $(p_{i}, \kappa_i)=(p_{i}, \sigma_i)$ that appears
in $\P(\sigma)\cap \P(\kappa)$. Since this holds for all graphs $G\in\bD$, it follows that 
the subcomplex $\P(\sigma)\cap \P(\kappa)$, as the union of the resulting identical maximal 
faces, has dimension at most $\dim(\rho)$. Now, since exactly the same reasoning also
applies when $\kappa$ is replaced by $\tau$, we get $\P(\sigma)\cap\P(\kappa)=
\P(\sigma)\cap\P(\tau)$, so \cref{eq:bc} and hence boundary consistency of $\P$ holds.

We can now just \emph{define} $\P(\rho) = \P(\sigma \cap \kappa) := \P(\sigma) \cap \P(\kappa)$,
which secures \cref{eq:Pstrict} for facets $\sigma, \kappa \in \K$. For general simplices,
assume for a contradiction that there are $\sigma$, $\kappa$ with $\rho=\sigma \cap \kappa\neq \emptyset$ but $\P(\rho) \neq 
\P(\sigma) \cap \P(\kappa)$. Choose facets $\sigma'$, $\kappa'$ and $\sigma''$, $\kappa''$ satisfying $\rho=\sigma' \cap \kappa'$, $\rho=\sigma'' \cap \kappa''$,
$\sigma=\sigma'\cap\sigma''$ and $\kappa=\kappa'\cap \kappa''$, which is always possible. 
Applying \cref{eq:Pstrict} to all
these pairs results in $\P(\rho) = \P(\sigma') \cap \P(\kappa') = \P(\sigma'') \cap \P(\kappa'')$, 
$\P(\sigma) = \P(\sigma') \cap \P(\sigma'')$
and $\P(\kappa) = \P(\kappa') \cap \P(\kappa'')$. We hence find
\[
\P(\rho) \neq \P(\sigma)\cap \P(\kappa) = \P(\sigma') \cap \P(\sigma'') \cap \P(\kappa') \cap \P(\kappa'') = \P(\rho) \cap \P(\sigma'') \cap \P(\kappa'') = \P(\rho),
\]
which is a contradiction. \qed
\end{proof}
Note that $\P$ can hence indeed be interpreted as a carrier map, according to \cref{def:carrmap},
which is strict. It is well known that strictness implies that, for any simplex $\rho \in \P(\K)$, 
there is a unique simplex $\sigma$ with smallest dimension in
$\K$, called the \emph{carrier} of $\rho$, such that $\rho \in \P(\sigma)$.

\medskip

A comparison with \cref{def:upcomplex} reveals that $\P_1=\P(\sigma_0)$ as given in
\cref{def:Pps} is indeed just the uninterpreted 1-round protocol complex.
The general $r$-round uninterpreted protocol complex $\P_r$, $r\geq 1$, is
defined as $\P(\P_{r-1})$, i.e., as the union of $\P$
applied to every facet $\sigma$ of $\P_{r-1}$, formally
$\P_r = \bigcup_{\sigma \in \P_{r-1}} \P(\sigma)$. 
Boundary consistency ensures that $\P_r=\P^r(\sigma_0)$ for the initial simplex
$\sigma_0 = \{(p_1,\{p_1\}), \dots, (p_n,\{p_n\})\}$  is well-defined for any $r\geq 0$.
An example for $r=2$ can be found in the
bottom part of \cref{fig:RAScomplexes}. 
Note that the arrows of the in-edges of a vertex $(p_i,h_r(p_i))$ in a facet in $\P_r$ 
represent the outermost level in \cref{eq:hr}; the labeling of the in-edges of $p_i$ in earlier
rounds $< r$ is no longer visible here. However, given the simplex 
$\rho=\{(p_1,h_r(p_1)), \dots, (p_n,h_r(p_n))\} \in \P_{r}$, the
labeling of the vertices $(p_j,h_{r-1}(p_j))\in V(\sigma)$
of the carrier $\sigma \in \P_{r-1}$ of $\rho$, i.e., the unique
simplex satisfying $\rho \in \P_1(\sigma)$, can be used to
recover the arrows for round $r-1$. 

We note that $\P_r=\P^{r-1}(\P(\sigma_0)) = \P(\P^{r-1}(\sigma_0))$ allows to view
the construction of $\P_r$ equivalently as applying the one-round construction $\P$ 
to every facet $F_{r-1}$ of $\P_{r-1}$ or else as applying the $(r-1)$-round construction $\P^{r-1}$
to every facet $F$ of $\P_1$. Boundary consistency of $\P$ again ensures
that this results in exactly the same protocol complex. Our decision procedure
for consensus solvability/impossibility provided in \cref{sec:consensuschar} 
will benefit from the different views provided by this construction.

\medskip

In the remainder of this section, we will discuss some consequences of
the fact that the carrier map corresponding to a general protocol complex 
construction operator $\P$ is always strict but need not be rigid (recall \cref{lem:bcP}).
This is actually a consequence of the asymmetry in the protocol complex
construction caused by graphs $\bD$ that do not treat all processes
alike. 

Consider the complete communication pseudosphere shown in 
\cref{fig:pseudospheres}, which corresponds to $\bD$ containing
\emph{all} possible graphs with $n$ vertices
(recall \cref{def:recps}). 
It does treat all processes alike, which also implies
that its outer border, which is defined by $\P(\partial \sigma_0)$ (see \cref{def:border} below),
has a very regular structure: For example, the four white and green 
vertices aligned on the bottom side 
of the outer triangle of \cref{fig:pseudospheres} are actually an instance of 
the 2-process communication pseudosphere shown in (\ref{eq:graph}). Its
corresponding carrier map is rigid.
By contrast, the protocol complex for the message adversary 
$\bD=\{G_1,G_2\}$ depicted by the two highlighted facets corresponding 
to the graphs $G_1$ and $G_2$ in \cref{fig:pseudospheres} has a
very irregular border shown in \cref{fig:bordern3all}.

\begin{figure}
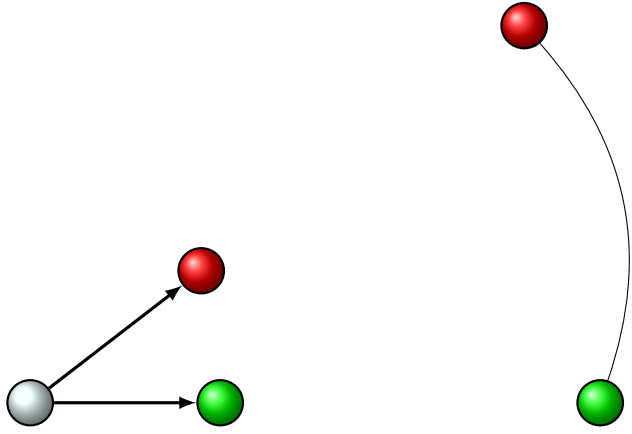

\ctikzfig{Figures/n3all_bent_border}
\caption{Border $\Bd(\P_1)$ of the simple message adversary shown in \cref{fig:pseudospheres}.}
\label{fig:bordern3all}
\end{figure}

It is worth mentioning, though, that there are other instances of
protocol complex construction operators that also have a rigid equivalent
carrier map. One important example is the popular standard chromatic subdivision~\cite{HKR13,Koz12:HHA},
which characterizes the iterated immediate snapshot (IIS) model of shared memory~\cite{HS99:ACT}:

\begin{definition}[Chromatic subdivision]\label{def:chrsub}
	Let $\mathcal{K}$ be an $(n-1)$-dimensional simplicial complex with a proper coloring $\chi: V(\mathcal{K}) \rightarrow \{p_1, \ldots, p_n\}$. We define the chromatic subdivision through its vertex set and facets as follows:
\begin{align}		
	V(\Ch(\mathcal{K})) =& 
	\bigl\{(p_i,\sigma) \in\Pi\times F(\K)\; | \;  p_i \in \chi(\sigma)\bigr\}, \label{eq:vertexCh}\\
	Fct(\Ch(\mathcal{K})) =& \bigl\{\{(p_1, \sigma_1), \ldots, (p_n, \sigma_n)\} \; | \;
	\exists\pi:[n]\to[n] \text{ permutation }
	\sigma_{\pi(1)} \subseteq \ldots \subseteq \sigma_{\pi(n)}, \nonumber\\
	&
	\forall 1 \leq i,j\leq n:
\chi(\sigma_i) \wedge (p_i\in\chi(\sigma_j) \Rightarrow \sigma_i \subseteq \sigma_j) \bigr\}.\label{eq:facetCh}
\end{align}	
\end{definition}

It is immediately apparent from comparing \cref{def:recps} and
\cref{def:chrsub} that $V(\Ps(\sigma^{n-1}))=V(\Ch(\sigma^{n-1}))$ and
$\Ch(\sigma^{n-1}) \subseteq \Ps(\sigma^{n-1})$, i.e., $\Ch(\sigma^{n-1})$ is 
indeed a subcomplex of $\Ps(\sigma^{n-1})$.
In \cref{fig:pseudospheres}, we have highlighted, via thick edges
and arrows, the protocol complex $\Ch(\sigma_0)$ for the corresponding 
message adversary.
In fact, the chromatic subdivision and hence the IIS model
is just a special case of our oblivious message adversary,
the set $\bD$ of which consists of all the directed graphs
that are unilaterally connected ($\forall G \in \bD, a,b\neq a \in V(G):
\mbox{$\exists$ directed path from $a$ to $b$ or from $b$ to $a$ in $G$}$)
and transitively closed ($\forall G \in \bD: (a,b),
(b,c)\in E(G) \Rightarrow (a,c)\in E(G)$).

\begin{lemma}[Equivalent message adversary for chromatic subdivision]\label{thm:MAforIIS}
	Let $\sigma_0$ be the uninterpreted input complex with process set $\Pi = \{p_1, \ldots, p_n\}$, and $\bD$ be the set of all unilaterally connected and transitively closed graphs on $\Pi$. Then, $\mathcal{P}(\sigma_0) = \Ch(\sigma_0)$.
\end{lemma}

\begin{proof}
	Notice first that for any face $\sigma \in \sigma_0$ such that $p_i \in \chi(\sigma)$, there exists a graph $G_\sigma \in \bD$ such that $\In^{G_\sigma}(p_i)=\chi(\sigma)$: simply consider $E(G_{\sigma}) = \{ (u,v) \; \vert \; u \in \chi(\sigma) \wedge v \in \Pi \} \cup \{ (w,y) \; \vert \; w, y \notin \chi(\sigma)\}$. By construction, $G_\sigma$ is both transitively closed and unilaterally connected. Therefore, $V(\Ch(\sigma_0)) \subseteq V(\mathcal{P}(\sigma_0))$. On the other hand, from \cref{def:Pps} of the protocol complex pseudosphere construction, it follows that $V(\mathcal{P}(\sigma_0))\subseteq V(\Ch(\sigma_0))$. Consequently, $V(\mathcal{P}(\sigma_0)) = V(\Ch(\sigma_0))$.
	
	Let $\sigma= \{(p_1, \sigma_1), \ldots, (p_n, \sigma_n)\}$ be a facet of $\Ch(\sigma_0)$, and consider the graph $G_\sigma$ with edges $E(G_\sigma) = \{(p_j,p_i) \; \vert \; p_j \in \chi(\sigma_i)\}$. Assume that $(p_i,p_j)$, $(p_j,p_k) \in E(G_\sigma)$. By \cref{def:chrsub}, it holds that $p_i \in \chi(\sigma_j)$, which implies $\sigma_i \subseteq \sigma_j$. Analogously, $\sigma_j \subseteq \sigma_k$ and therefore $\sigma_i \subseteq \sigma_k$. It hence follows that $p_i \in \chi(\sigma_k)$, and by construction of $G_\sigma$, that $(p_i,p_k) \in E(G_\sigma)$. This shows that $G_\sigma$ is transitively closed.
	
	 Now consider $p_i, p_j \in \Pi$. Since $\pi$ is a permutation, $p_i = p_{\pi(i')}$ and $p_j = p_{\pi(j')}$ for some $i',j' \in [n]$. Let us assume w.l.o.g that $i' \leq j'$. Then $\sigma_i = \sigma_{\pi(i')} \subseteq \sigma_{\pi(j')}= \sigma_j$, which implies that $p_i \in \chi(\sigma_j)$. From the definition of $G_\sigma$, it follows that $(p_i,p_j) \in E(G_\sigma)$. This shows that $G_\sigma$ is also unilaterally connected. Therefore, $\sigma$ must also be a facet of $\mathcal{P}(\sigma_0)$, i.e., $Fct(\Ch(\sigma_0)) \subseteq Fct(\mathcal{P}(\sigma_0))$.
	
	Conversely, let $\sigma=\{(p_1, \sigma_1), \ldots, (p_n, \sigma_n)\}$ be a facet of $\mathcal{P}(\sigma_0)$. Let $G_\sigma$ be the graph from $\bD$ that induces $\sigma$. Recall that $G_\sigma$ is unilaterally connected and transitively closed. Let $S_i$ denote the strongly connected component containing $p_i$. Since $G_\sigma$ is transitively closed, $S_i$ is in fact a directed clique. Therefore, $S_i \subseteq \chi(\sigma_i)= \In^{G_\sigma}(p_i)$. 
	Consider the component graph $G^*$ where $V(G^*) = \{ S_i \mid i \in [n]\}$, and $E(G^*) = \{(S_i,S_j ) \; \vert \; (p_i,p_j) \in E(G_\sigma) \}$. Since $G_\sigma$ is transitively closed and unilaterally connected, $G^*$ is a transitive tournament (where $(a,b)$ or $(b,a)$ must be present for all $a,b\in V(G^*)$).
        Therefore, $G^*$ has a directed Hamiltonian path $S_{\pi(1)}, \ldots, S_{\pi(s)}$ for $s=|V(G^ *)|$; note that $s\leq n$ since $S_i=S_j$ may be the same for different processes $p_i$ and $p_j$.
	
	Clearly, the permutation from the Hamiltonian path of connected components, extended by ordering processes leading to the same connected component according to their ids, induces a complete ordering of the process indices: $i \preceq j$ if $S_i = S_{\pi(i')}$ and $S_j = S_{\pi(j')}$ with $i' \leq j'$ and $i \leq j$, i.e., first we order each index $i$ according to the order of their connected component in the Hamiltonian path in $G^*$, and break ties according to their process ids. Therefore, $\preceq$ is a total ordering on $[n]$, and thus induces a permutation $\pi'$ with the property that if $i \leq j$, then either $S_{\pi'(i)} = S_{\pi'(j)}$, or there exists an edge from $S_{\pi'(i)}$ to $S_{\pi'(j)}$.
	
	From the transitive closure of $G_\sigma$ and the construction of $\pi'$, we get $\In^{G_\sigma}(\pi'(p_i)) = \bigcup \limits_{j=1}^{i} S_{\pi'(i)}$. Therefore, $\pi'$ is also a permutation of the $\sigma_i$ in $\sigma$ that satisfies the conditions for being a facet of $\Ch(\sigma_0)$. It follows that $\sigma \in Fct(\Ch(\sigma_0))$. Therefore $Fct(\mathcal{P}(\sigma_0)) \subseteq Fct(\Ch(\sigma_0))$, which completes the proof that $Fct(\mathcal{P}(\sigma_0)) = Fct(\Ch(\sigma_0))$ and thus $\mathcal{P}(\sigma_0) = \Ch(\sigma_0)$. \qed	
\end{proof}

For any pair of simplices $\sigma,\kappa \in \K$, it hence holds by \cref{eq:Pstrict}
that $\Ch(\sigma) \cap \Ch(\kappa) = 
\Ch(\sigma \cap \kappa)$, i.e., subdivided simplices  
that share a face intersect precisely in the subdivision of that
face in $\Ch(\K)$. \cref{lem:bcP} thus ensures that the iterated standard chromatic 
subdivision $\Ch^r(\K)$ is well-defined. 

Thanks to its regular structure, the equivalent carrier map is also rigid.
As is the case for the communication pseudosphere in \cref{fig:pseudospheres}, 
the four white and green vertices aligned on the bottom side of the outer triangle 
connected by thick arrows are actually an instance of 
a 2-process chromatic subdivision. Indeed, the standard chromatic subdivision $\Ch(\sigma^\ell)$ 
of a simplex $\sigma^\ell$ of dimension $\ell$ can be constructed iteratively \cite{HKR13}: Starting out from the vertices
$V(\sigma^\ell)$, i.e., the 0-dimensional faces $\sigma^0$ of $\sigma^\ell$, where 
$\Ch(\sigma^0)=\sigma^0$, one builds $\Ch(\sigma^1)$ for the edge $\sigma^1$
by placing 2 new vertices in its interior and connecting them to each other and to the
vertices of $\sigma^1$. For constructing $\Ch(\sigma^3)$, one places 3 new vertices
in its interior and connects them to each other and to the vertices constructed before, 
etc. 

\begin{corollary}
	Let $\mathcal{K}$ be an arbitrary simplicial complex, then $\Ch(\mathcal{K}) = \mathcal{P}(\mathcal{K})$ with $\bD$ as the set of allowed graphs.
\end{corollary}
\begin{proof}
	Follows immediately from \cref{thm:MAforIIS} and boundary consistency of $\mathcal{P}(\mathcal{K})$.
\end{proof}

\subsection{Classification of facets of protocol complexes}

We first define the important concept of the border of a protocol complex.

\begin{definition}[Border]\label{def:border}
The \emph{border} $\Bd(\P_1)$ of a 1-round protocol complex $\P_1=\P(\sigma_0)$ 
is defined as $\Bd(\P_1)=\P(\partial \sigma_0)$. The \emph{border} $\Bd(\P_r)$
(resp.\ the border $\Bd(\C)$ of some subcomplex $\C \subseteq \P_r$) of the general 
$r$-round complex $\P_r=\P^r(\sigma_0)$ is $\Bd(\P_r)=\P^r(\partial \sigma_0)$.
\end{definition}

Due to the boundary consistency property of $\P$ (\cref{lem:bcP}), 
the border is just the ``outermost'' part of $\P_r$, i.e., the part
that is carried by $\partial \sigma_0$; the dimension of every 
facet $F\in\Bd(\P_r)$ is at most $\dim(\sigma_0)-1=n-2$. 
Recall that it may also be smaller than $n-2$, since $\P$ viewed as a carrier
map need not be rigid. Obviously, however, $F$ is always a face of some facet in
$\P_r$. In the case of \cref{fig:pseudospheres}, 
	where $\P_{1}=\{G_1,G_2\}$
	with the graphs $G_1,G_2$ given in Eqn.~(\ref{eq:advcomm}),
$\Bd(\P_1)$ only consists of 
the three edges and the vertices shown in
\cref{fig:bordern3all}.
Observe that the processes of the vertices $V(\rho)$ of a face $\rho \in \Bd(\P_r)$
may possibly have heard from each other, but not from processes in $\Pi\setminus V(\rho)$, 
in any round $1,\dots, r$. 

\begin{figure}
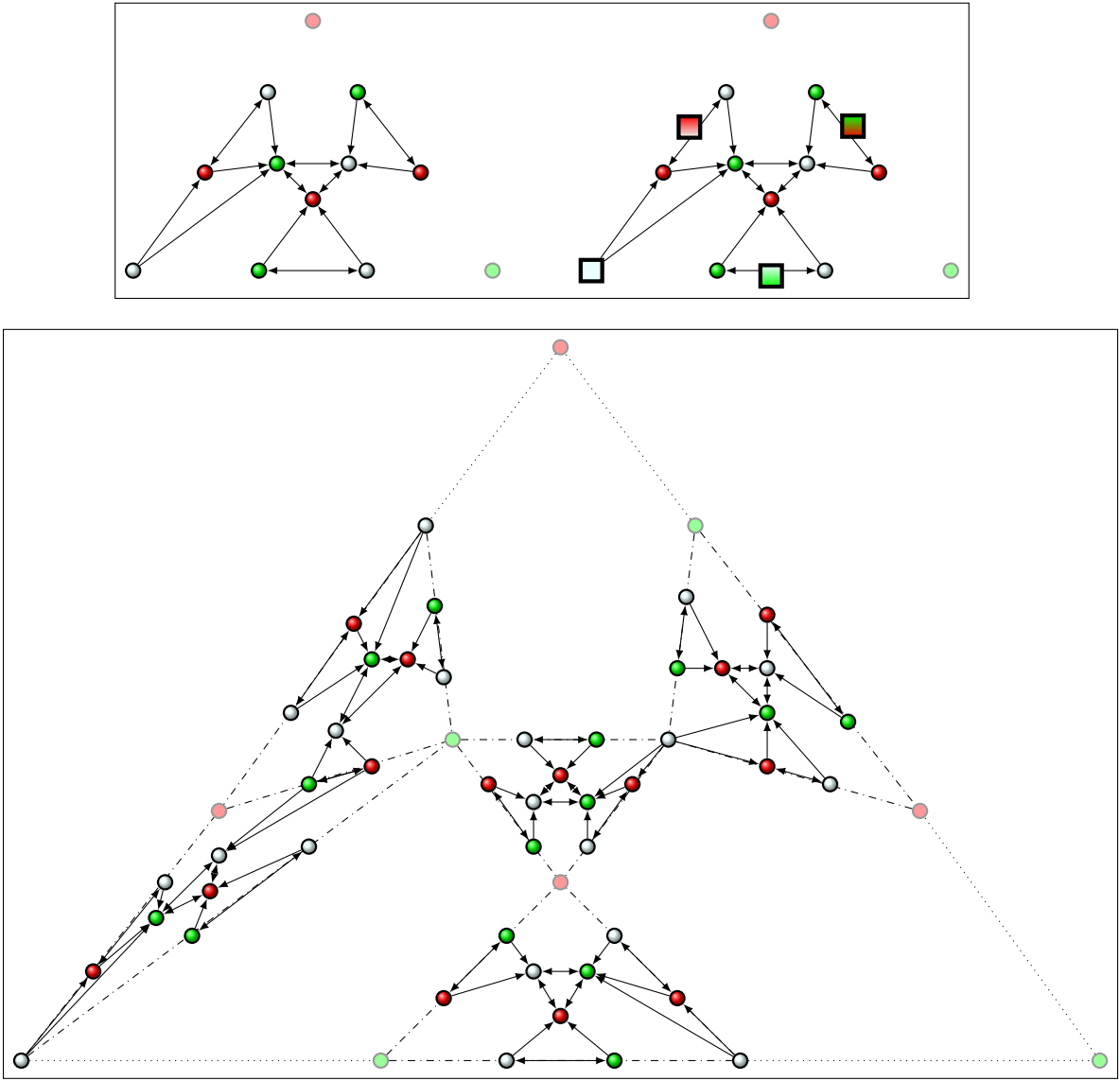

\ctikzfig{Figures/ras1}
\ctikzfig{Figures/ras2}
\caption{Protocol complex $\P_1^{RAS}$ for one round (top) and $\P_2^{RAS}$ for two rounds (bottom) of the RAS message adversary. The top right figure also shows the border root components of $\P_1^{RAS}$.}
\label{fig:RAScomplexes}
\end{figure}

For a more elaborate running example, consider the RAS message adversary 
shown in \cref{fig:RAScomplexes}: 
Its 1-round uninterpreted complex $\P_1^{RAS}$ (top left part) 
is reminiscent of the well-known radioactivity sign, hence its name.
Its 2-round uninterpreted complex $\P_2^{RAS}$ is shown in the bottom part of the figure. It is 
constructed by taking the union of the 1-round uninterpreted complexes $\P(F)$ for every facet 
$F \in \P_1$. Its border $\Bd(\P_2^{RAS})$ is formed by all the vertices
and edges of the faces that lie on the (dotted and partly dash-dotted) borders of the outermost triangle.

For classifying the facets in a protocol complex, the \emph{root components} of the graphs 
in $\bD$ will turn out to be crucial. 

\begin{definition}[Root components]\label{def:root}
  Given any facet $F$ in the protocol complex $\P_r$, $r\geq 1$, let $\sigma \in \P_{r-1}$ be its carrier, i.e., the unique facet such that $F\in \P(\sigma)$, and $G\in \bD$ be the corresponding graph leading to $F$ in $\P(\sigma)$.
  A \emph{root component} $R(F)$ of $F$
  is the face of $F$ corresponding to a strongly connected component $R$ in $G$
  without incoming edges from $G\setminus{R}$. 
\end{definition}

It is well-known that every directed graph with $n$ vertices has at least one and 
at most $n$ root components, and that every process in $G$ is reachable from every
member of at least one root component via some directed path in $G$. 
Graphs with a single root component are called rooted, and
it is easy to see that just one graph in $\bD$ that is not rooted makes consensus trivially 
impossible:
The adversary might repeat this graph forever, preventing the different root components from coordinating the output value.
In the sequel, we will therefore restrict our attention to message
adversaries where $\bD$ is made up of rooted graphs only, and will 
denote by $R(G)=R(F)$ the face representing the root component of $F$. Note that $R(G)$ is a
face and hence includes the edges of the interconnect and their orientation; 
its set of vertices is denoted by $V(R(F))=\{(p_i,h_r(p_i))\;|\; p_i \in \chi(R(F))\}$.
Recall from \cref{def:HOhistory} that the set of processes that $p_i$ has actually heard of in some 
vertex $v=(p_i,h_r(p_i))\in V(\rho)$ is denoted $\cup h_r(p_i)$.

\begin{definition}[Border facets]\label{def:borderfacets}
A facet $F\in \P_r$ is a \emph{border facet}, if the subcomplex $F \cap \Bd(\P_r)$ 
is non-empty. 
The subcomplex $F \cap \Bd(\P_r)$ will be called \emph{facet borders} of $F$. 
A border facet $F$ is \emph{proper}
if the members of the root component did not collectively hear from all processes, i.e., 
$\bigcup_{(p_i,h_r(p_i)) \in V(R(F))} \cup h_r(p_i) \neq \Pi$.
\end{definition}

Intuitively, a border facet $F\in \P_r$ has at least one vertex $v \in \Bd(\P_r)$.
It is immediately apparent that $v$ may have heard at most from processes in some 
face $\rho\in\Bd(\P_r)$, which has dimension at most $n-2$, but not from
processes outside $\rho$ (so, in particular, not from all processes).

The facet borders $F \cap \Bd(\P_r)$ of a border facet $F$ form indeed a subcomplex in general, rather than
a single face, as is the case in, e.g., the left part of \cref{fig:bordern3all} (generated by $F$ that represents the graph $G_2$) shows. 
Moreover, $F\cap\Bd(\P_r)$ does not even
need to be connected. 
For example, if the message adversary of \cref{fig:bordern3all} would also include the graph $G_3=\{r \to g \to w\}$, i.e., a chain (with root component $r$), we observe $F_3\cap\Bd(\P_1) = \{r \to g, w\}$ 
for the corresponding facet $F_3$. 
Finally, it need not
even be the case that $F \cap \Bd(\P_r)$ contains the entire
root component $R(F)$: Since $\dim(\Bd(\P_r))=n-2$, this is inevitable
if $F$ is not a proper border facet, i.e., if the members of $R(F)$ have collectively heard from all processes. 
For instance, if the message adversary of \cref{fig:bordern3all} also contained the cycle $G_4=\{r\to g \to w \to r\}$ (with root component $R(F_4)=F_4=\{r\to g \to w \to r\}$ consisting of all processes), 
then the (improper) border facet $F_4\cap\Bd(\P_1) = \{r, g, w\}$ obviously cannot contain $R(F_4)$.

\begin{definition}[Border components and border root components]\label{def:bcrc}
For every proper border facet $F \in \P_r$, the \emph{border component} $B(F)$ is
the smallest face of $F$ whose members did not hear from processes outside of $B(F)$, that is, $\bigcup_{(p_i,h_r(p_i))\in V(B(F))} \cup h_r(p_i) 
\subseteq \chi(B(F))$. 
For a facet $F$ that is not a proper border facet, we use the convention $B(F)=F$ for completeness.
The set of all proper border components of $\P_r$ is denoted as $\BdC(\P_r)$ 
(with the appropriate restriction $\BdC(\C)$ for a subcomplex $\C\subseteq \P_r$).

The root component $R(F)$ of a proper border facet $F$ is called
\emph{border root component}; it necessarily satisfies $R(F)\neq F$.
The set of all border root components of $\P_r$ resp.\ a 
subcomplex $\C\subseteq \P_r$ is denoted $\BdR(\P_r)$ resp. $\BdR(\C)$.
\end{definition}

\Cref{lem:propbc} below will assert that the border component of a facet is unique and contains its root component.

\begin{definition}[Border component carrier]\label{def:bccarr}
The \emph{border component carrier} $\B(F)$ of a proper border facet $F$ is the 
smallest face of the initial simplex $\sigma_0=\{(p_1,\{p_1\}), \dots, (p_n,\{p_n\})\}$ that carries $B(F)$.
For a facet $F$ that is not proper, we use the convention $\B(F)=\sigma_0$ for consistency.
\end{definition}
 
Since $\chi\bigl(\B(F)\bigr)=\chi(B(F))$, it is apparent that $\B(F)$ implicitly also characterizes the heard-of sets of the processes in $B(F)$: According to \cref{def:bcrc}, its members may have heard from processes in $\B(F)$ but not from other processes. Note carefully that this also tells something about the knowledge of the processes regarding the initial values of other processes, as the members of $B(F)$ can at most know the initial values of the processes in $\B(F)$.

For an illustration, consider the top right part of \cref{fig:RAScomplexes},
which shows the border root components of border facets of $\P_1^{RAS}$ for the RAS message adversary, represented by square nodes with fat borders. 
$\BdR_2^{RAS}$ depicted in the bottom part of \cref{fig:RAScomplexes}
consists of
all faces $B(F)$ of border facets $F$ touching the outer border: Going in clockwise direction, starting with the 
bottom-leftmost border face, we obtain the following pairs 
(border root component, border component carrier) representing $B(F)$ of a border facet: $(w,\{w\})$, $(r \leftrightarrow w, \{r,w\})$, $(r \leftrightarrow w, \{r,w\})$,  $(w,\{w,r\})$, $(r \leftrightarrow g, \{r,g\})$, $(w,\{w,g\})$, $(w \leftrightarrow g, \{w,g\})$. It is apparent that the members of the border root
component $w \leftrightarrow g$ in the last pair $(w \leftrightarrow g, \{w,g\})$ (representing the border facet on the bottom) only know their own initial values, but not the initial value of the red process.

\begin{lemma}[Properties of border component of a proper border facet]\label{lem:propbc}
The border component $B(F)$ of a proper border facet $F \in \P_r$ satisfies the following properties:
\begin{enumerate}
\item[(i)] $B(F)$ is unique,
\item[(ii)] $R(F) \subseteq B(F) \subseteq F \cap \Bd(\P_r)$, which also implies $B(F)\neq F$, 
\item[(iii)] $B(F)=R(F)$ for $r=1$, but possibly $R(F) \subsetneq B(F)$ for $r \geq 2$.
\end{enumerate}
\end{lemma}
\begin{proof}
As for (i), assume for a contradiction that there is some alternative $B'(F)$ of the same size.
Due to \cref{def:root}, both $R(F) \subseteq B(F)$ and $R(F) \subseteq B'(F)$ must hold, 
since some process in $B(F)$ would have heard from a process in $R(F)\setminus B(F)$ otherwise,
and, analogously, for $B'(F)$.
As $B(F)\neq B'(F)$, there is a $v'=(p_i',h_r(p_i'))  \in B'(F)\setminus B(F)$ that is present in $B'(F)$ because some $v\in B(F) \cap B'(F)$ has heard from $p_i'$ earlier.
But then, $v'$ is also present in $B(F)$, a contradiction.
For (ii), besides $R(F) \subseteq B(F)$, we also have
$\sigma = R(F) \cap \Bd(\P_r) \neq \emptyset$, since $R(F)$ of a proper border facet 
according to \cref{def:borderfacets} does not encompass all processes.
Now assume first that $R(F) \not\subseteq F \cap 
\Bd(\P_r)$, i.e., $R(F) \not\subseteq \Bd(\P_r)$ (since $R(F) \subseteq F$ obviously always holds).
For every facet $\sigma \in \Bd(\P_r)$, there is hence some $v=(p_i,h_r(p_i)) \in R(F)\setminus \sigma \neq \emptyset$. 
However, by the properties of root components, some process $p_j \in \chi(\sigma)$ 
must have heard from $p_i \not\in \chi(\sigma)$, which would contradict $p_j \in \chi(\sigma)$. 
Therefore, we must have $R(F) \subseteq \Bd(\Pr)$. For the final contradiction, by the same token,
assume that $B(F)\setminus\Bd(\P_r) \neq \emptyset$, i.e., for any facet $\sigma \in \Bd(\P_r)$, 
there is some $v=(p_i,h_r(p_i)) \in B(F)\setminus\sigma$. By the definition of border components according to
\cref{def:bcrc}, however, such a $v$ exists only if some process $p_j \in \chi(\sigma)$ has already heard
from $p_i \not\in \chi(\sigma)$, which would contradict $p_j \in \chi(\sigma)$. Thus, 
$B(F) \subseteq F \cap \Bd(\P_r)$. Finally,
$B(F)\neq F$ follows from the fact that $B(F) \subseteq \Bd(\P_r)$ imposes a maximum dimension
of $n-2$ for $B(F)$.

As for (iii), $B(F)=R(F)$ for $r=1$ follows immediately from \cref{def:bcrc}. For $r\geq 2$, it is of course possible that some process in $R(F)$ has already heard from a process outside 
$R(F)$ in some earlier round, see the bottom-left facet $F_2^l$ in \cref{fig:bordercomponents} for an example.
\qed
\end{proof}

We point out that the facet borders 
$F \cap \Bd(\P_r)$ of a proper border facet $F \in \P_r$ for $r\geq 2$ may also contain 
(small) faces in $\Bd(\P_r)$ that are disjoint
from $B(F)$, albeit they are of course always contained in a larger face of $F$ that also
contains $B(F)$. Examples can be found in the left part of \cref{fig:bordern3all}, where
$B(F)=\{w\}$ is disjoint from the single vertex $\{g\} \in F \cap \Bd(\P)$, or in the top-left facet $F_1^l$ in \cref{fig:bordercomponents}, where $B(F_1^l)=w$ is disjoint from the single vertex $\{g\} \in F_1^l \cap 
\Bd(\P_1)$.

\medskip

We conclude this section by stressing that border facets and border (root) components in $\P_r$ for $r\geq 2$
implicitly represent \emph{sequences} of faces: A single border facet $F$ and the
corresponding $R(F)\in \BdR(\P_r)$ resp.\ $B(F)\in \BdC(\P_r)$ actually represent 
all the border facets $F_1,\dots,F_r$ in rounds $1,\dots, r$ 
that carry $F=F_r$, and their corresponding $R(F_1),\dots,R(F_r)$ resp.\ $B(F_1),\dots, B(F_r)$.
Note carefully that, although $\B(F_{i+k})$, $k\geq 1$, is typically not
smaller than $\B(F_i)$, in particular, if $\chi\bigl(R(F_{i+k})\bigr)=\chi\bigl(R(F_i)\bigr)$, this need not
always be the case, since a process present in $B(F_i)$ may have heard from all 
other processes in $F_{i+k}$, so that it is no longer present in
$\Bd(F_{i+k})$ and hence in $\B(F_{i+k})$.
\cref{fig:bordercomponents} provides some additional illustrating examples
of border component carriers and border root components. 

\begin{figure}
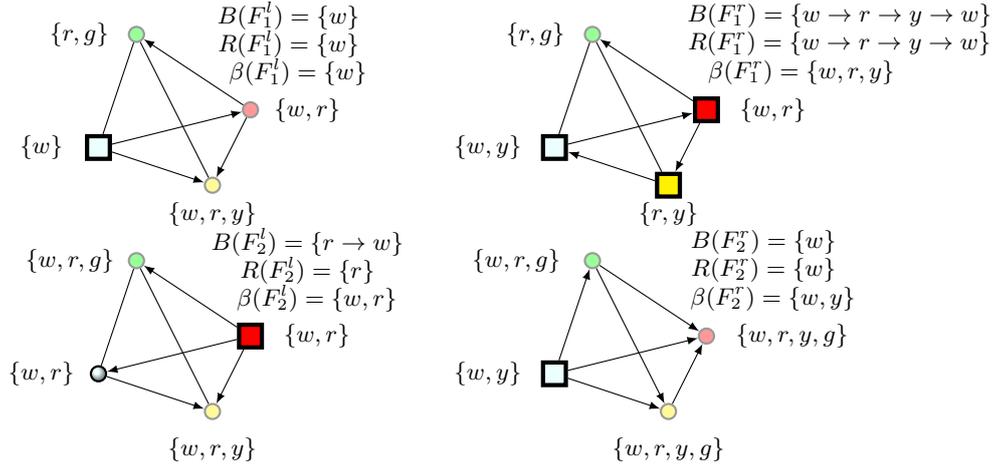

\ctikzfig{Figures/bordercomponents}
\caption{Illustration of border components and border root components, for two different examples (left, right column) for $n=4$ processes. Faded nodes represent vertices outside $B(F)$; squared nodes represent members of the border
root component. The first row shows the respective border facets 
in $\P_1$, the second row shows the border facets in $\P_2$. The labels provide the heard-of set of the nearby
process, assuming that the round-2 facet is applied to the round-1 facet atop of it. Observe that the red and yellow process in the bottom-right facet $F_2^r$ 
have heard from everybody and are hence removed from $\Bd(F_2^r)$ and $\B(F_2^r)$.}
\label{fig:bordercomponents}
\end{figure}

\section{Consensus Solvability/Impossibility}
\label{sec:consensuschar}

In this section, we will characterize consensus solvability/impossibility under an oblivious 
message adversary $\MA$ by means of the topological tools introduced in \cref{sec:topologyintro}.
Due to its ``geometrical'' flavor, our topological view not only provides interesting additional insights, 
but also prepares the ground for additional results provided in \cref{sec:terminationtime}.

The key insight of \cref{sec:incompbc} is that one cannot solve consensus in $r$ rounds if the $r$-round protocol 
complex $\P_r$ comprises a connected component that contains \emph{incompatible} proper 
border facets (defined as having a set of border components with an empty intersection).
In \cref{sec:charconsensus}, we focus on paths connecting pairs of incompatible proper
border facets in $\P_{r-1}$, and exhaustively characterize what happens to such a path
when $\P_{r-1}$ evolves to $\P_r$: It may either break, in which case consensus might
be solvable in $\P_r$ (unless some other path still prevents it), or it may be lifted, 
in which it still prohibits consensus. In \cref{sec:charconnectedcomp}, we recast our
path-centric characterization in terms of its effect on the connected components
in the evolution from $\P_{r-1}$ to $\P_r$. A suite of examples in \cref{sec:examples}
illustrates all the different cases. Finally, \cref{sec:decisionprocedures} presents a
alternative (and sometimes more efficient) formulation of the consensus decision 
procedure given in \cite{WPRSS23:ITCS},
which follows right away from the topological characterization of consensus solvability
developed in the previous subsections.

\subsection{Incompatibility of border components}
\label{sec:incompbc}

Consensus is impossible to solve in $r$ rounds
if the $r$-round protocol complex $\P_r$ has a connected component $\C$ that contains $k\geq 1$ proper border facets $\hF_1,\dots,\hF_k$ with incompatible border components $B(\hF_1),\dots, B(\hF_k) \in \Bd(\P_r)$, where incompatibility means $\bigcap_{i=1}^k \B(\hF_i) = \emptyset$:
Since no vertex of $B(\hF_i)$ could have had incoming edges from processes outside $B(\hF_i)$, in any of the rounds $1,\dots,r$, their corresponding processes cannot decide on anything but one of their own initial values. 
As all vertices in a connected component must decide on the same value, however, this is impossible. 

Incompatible border components occur, in particular, when $\hF_1,\dots,\hF_k$ have incompatible border root components $R(\hF_1),\dots, R(\hF_k) \in \BdR(\P_r)$.
An instance of this situation can be seen in the top right part of \cref{fig:RAScomplexes}: 
Since there is a path from the bottom-left white vertex (shown as a fat squared node that
represents the border root component consisting only of this vertex) to the border root component
consisting of the red and green square on the right edge of the outer triangle, consensus
cannot be solved in one round.

\subsection{Characterizing solvability via paths connecting incompatible border components}
\label{sec:charconsensus}

In this subsection, we will characterize the possible evolutions of a path that connects facets with
incompatible border components in some protocol complex $\P_{r-1}$, for some $r\geq 1$, which may either
break or may lead to a \emph{lifted} path connecting incompatible border components in $\P_r$.

Consider two border facets $\hF_x\neq\hF_y$ taken from a  set
of $k\geq 2$ incompatible proper border facets $\hF_1,\dots,\hF_k \in \C \subseteq \P_{r-1}$, $r\geq 2$, i.e., 
belonging to the connected component $\C$ and having incompatible 
border components $B(\hF_1),\dots,B(\hF_k)$ (see \cref{fig:pr1} for an illustration).
Since $\P_{r-1}$ is the result of repeatedly applying $\P$ to the single facet $\sigma_0$, 
there must be some smallest round number $1\leq \orr \leq r-1$ and two facets
$\otau_x\neq \otau_y$ 
with $B(\otau_x)\neq B(\otau_y)$ 
in $\P_{\orr}$
that
carry $\hF_x$ and $\hF_y$, respectively, 
i.e., $\hF_x \in \P_{r-1-\orr}(\otau_x)$
and $\hF_y \in \P_{r-1-\orr}(\otau_y)$.
Note that, as $\orr$ is minimal,
$\otau_x$ and $\otau_y$ are facets obtained by applying
$\P$ to the same facet $\oF\in \P_{\orr-1}$
(see \cref{fig:p0p1}).
For simplicity of exposition, we will assume
below that $\orr=1$, as otherwise we would have to introduce the definition
of a ``generalized border'' that does not start from $\P_1$ but
rather from $\P_{\orr}$. We will hence subsequently just write
$\P_1$, $\P_{r-1}$ and $\P_{r-2}$
instead of $\P_{\orr}$, $\P_{r-\orr}$, and $\P_{r-1-\orr}$, respectively.
Fortunately, this assumption can be made without loss of generality,
as all the scenarios that can occur in the case of $\orr > 1$ will
also occur when $\orr=1$.

\begin{figure}
	\includegraphics[scale=0.6]{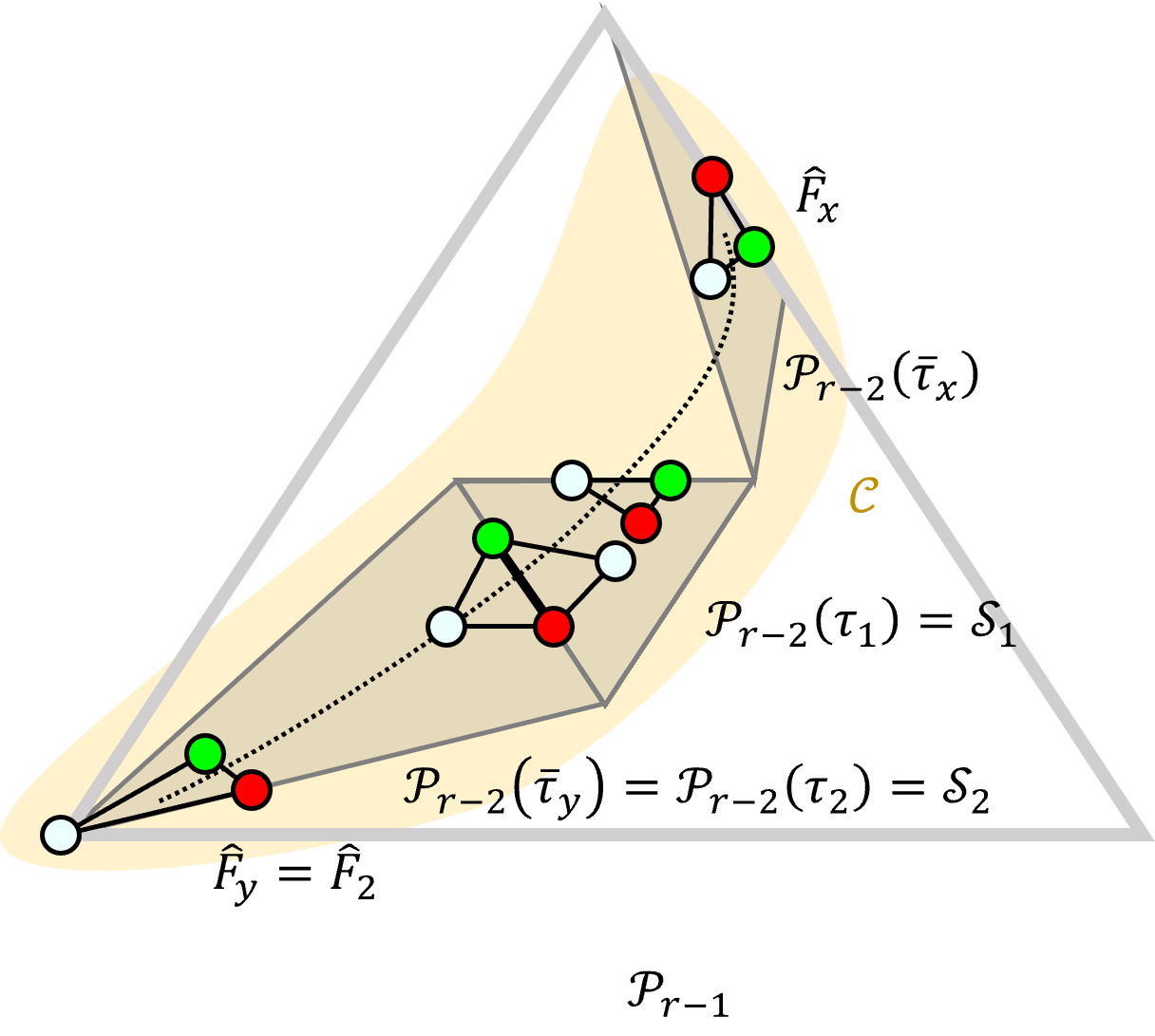} \hfil
	\includegraphics[scale=0.6]{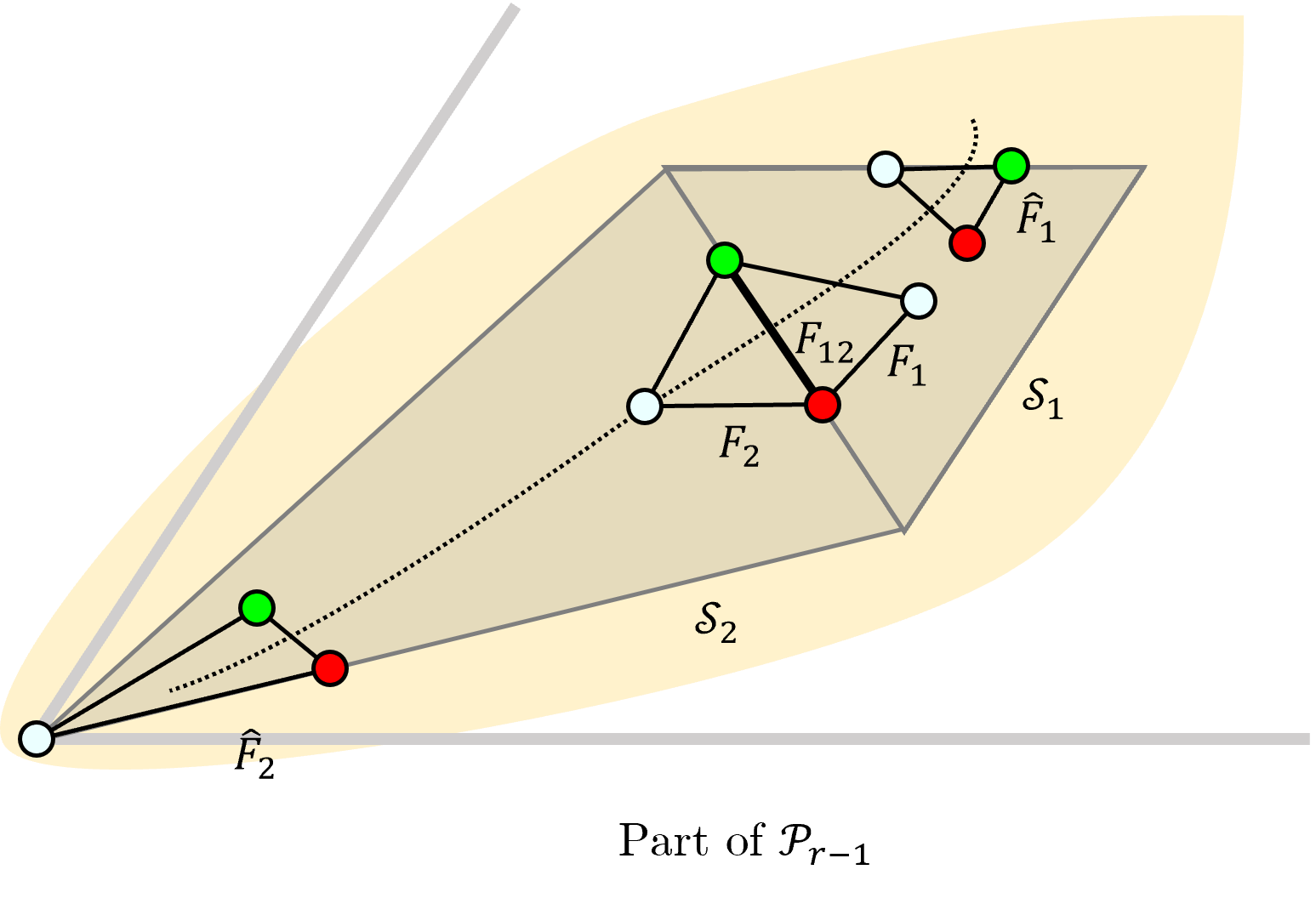}
	\centering
	\caption{$P_{r-1}$}
	\label{fig:pr1}
\end{figure}

\begin{figure}
	\includegraphics[scale=0.7]{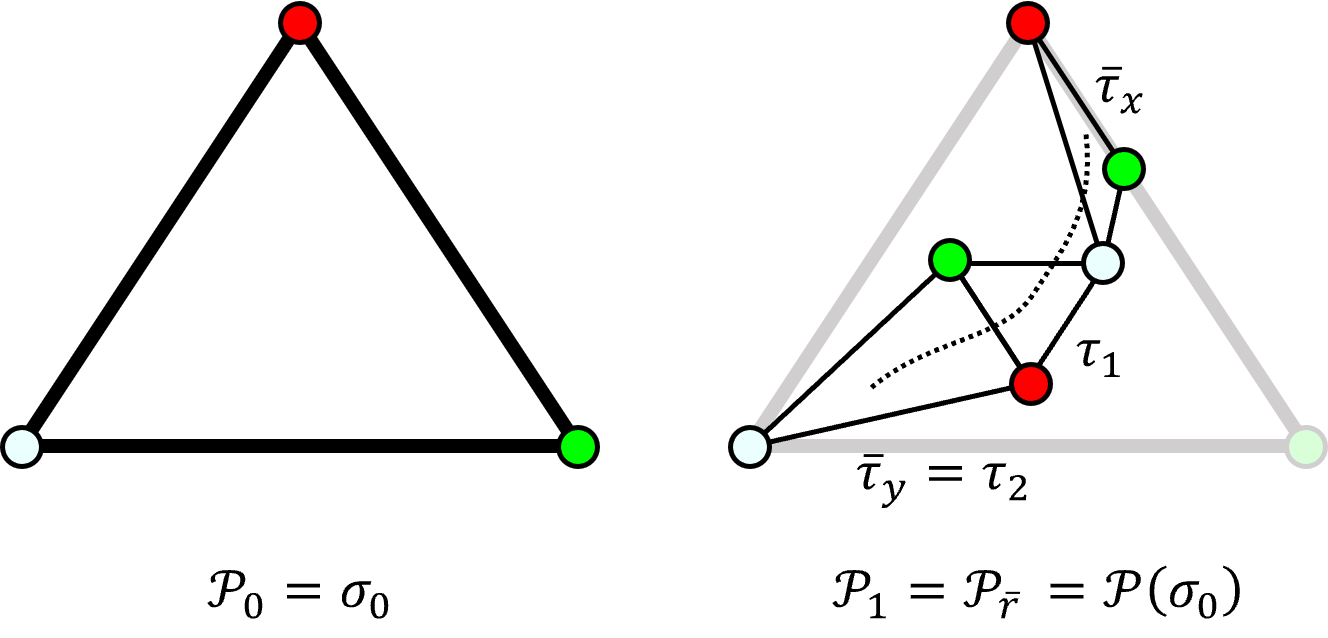}
	\centering
	\caption{$\P_0$ and $P_1$}
	\label{fig:p0p1}
\end{figure}

Since $\hF_x$ and $\hF_y$ are connected in $\C \subseteq \P_{r-1}$, 
$\otau_x$ and $\otau_y$
must be connected via one or more paths of adjacent facets in $\P_{1}$ as well. 
Consider an arbitrary, 
fixed path connecting the proper border facets $\hF_x$ and $\hF_y$ in $\P_{r-1}$, and its unique
corresponding path connecting $\otau_x$ and $\otau_y$ in $\P_{1}$. Let $\tau_1$ and $\tau_2$ be any two adjacent facets on the path in $\P_1$,
and $\tau_{12}=\tau_1 \cap \tau_2 \neq \emptyset$. In $\P_{r-1}$, the facets $\tau_1$ and $\tau_2$
induce connected subcomplexes $\S_1=\P^{r-2}(\tau_1)$ and $\S_2=\P^{r-2}(\tau_2)$ with a non-empty intersection $\S_1 \cap \S_2 \neq\emptyset$. 
The path from $\hF_x$ to $\hF_y$ in $\P_{r-1}$ must enter $\S_1$ at some facet $\hF_1$ and exit  $\S_2$ at some facet $\hF_2$, 
i.e., there is a path connecting $\hF_x$ to $\hF_1$ and a path connecting $\hF_y$ to $\hF_2$, 
and cross the border between $\S_1$ and $\S_2$ via adjacent facets $F_1 \in \S_1\subseteq \P_{r-1}$ and $F_2 \in \S_2\subseteq \P_{r-1}$ with
$\emptyset \neq F_{12}=F_1 \cap F_2 \in \S_1 \cap \S_2$. 
Note that $F_{12}$, as the intersection of two facets in the protocol complex $\P_{r-1}$, is of course a face. 

\begin{figure}
	\includegraphics[scale=0.7]{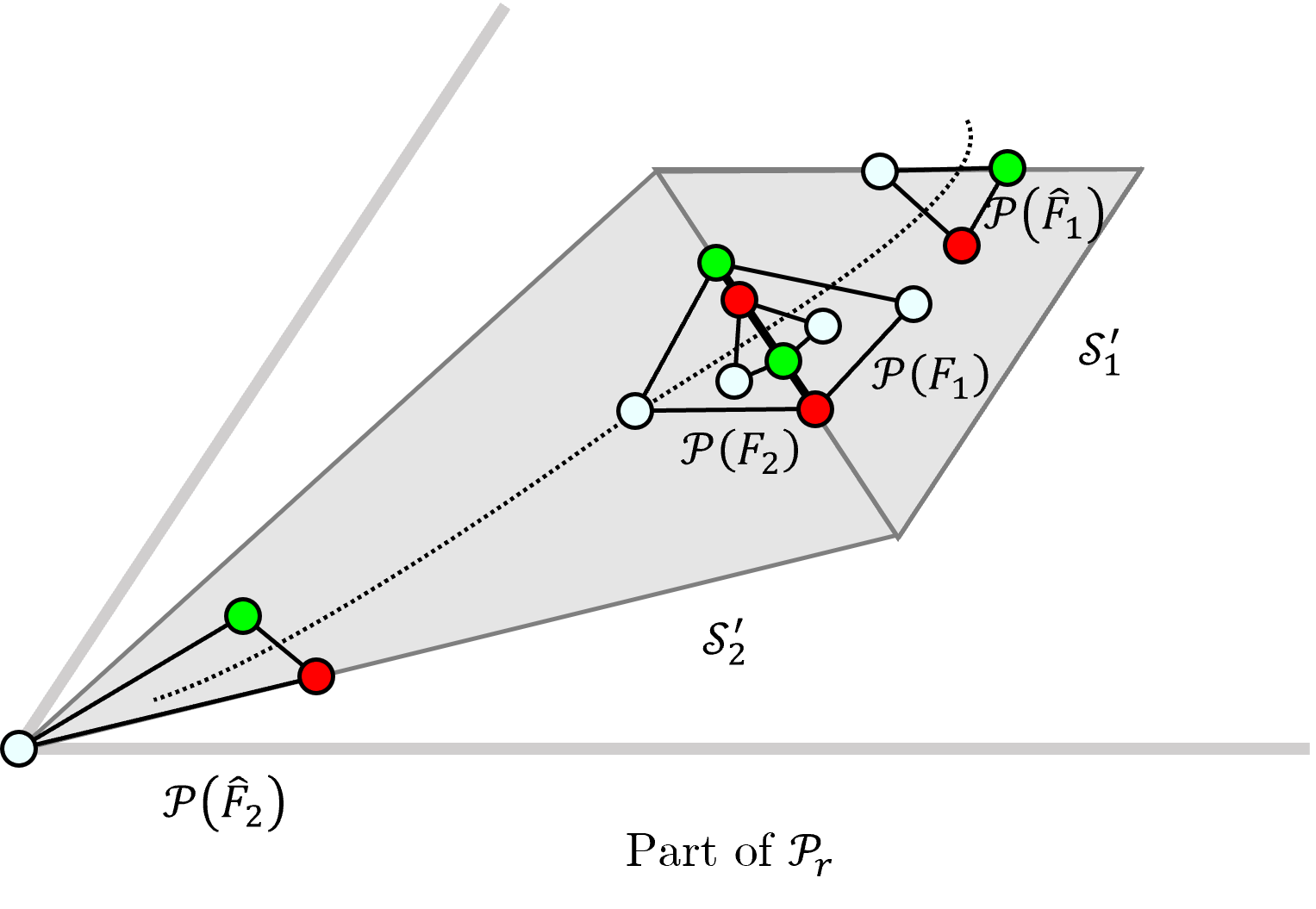}
	\centering
	\caption{Part of $\P_r$}
	\label{fig:prpartof}
\end{figure}

Now consider the outcome of applying $\P$ again to all the facets in $\P_{r-1}$, which of course gives~$\P_r$ (see \cref{fig:prpartof}).
In particular, this results in the subcomplexes $\S_1'=\P\bigl(\P^{r-2}(\tau_1)\bigr)=\P^{r-1}(\tau_1)$ and
analogously $\S_2'=\P^{r-1}(\tau_2)$, which may or may not have a non-empty intersection.
We will be interested in the part of this possible intersection created by the application of $\P$ to $F_1$ and $F_2$, i.e., in
$\P(F_1) \cap \P(F_2) \subseteq \S_1' \cap \S_2'$. 
Note that both $\S_1'$ and $\S_2'$ are isomorphic to $\P_{r-1}$. 
Clearly, the application of $\P$ to the facets $F_1, F_2 \in \P_{r-1}$
typically creates many pairs of intersecting border facets $F_1' \in \P(F_1) \subseteq \S_1'\subseteq \P_r$ and
$F_2' \in \P(F_2) \subseteq \S_2' \subseteq \P_r$, such that each pair shares some
non-empty face $\emptyset\neq F_{12}'=F_1'\cap F_2' \subseteq \S_1' \cap \S_2'$. 
The shared faces $F_{12}'$ together form the subcomplex $FC_{12}' \subseteq \P(F_1) \cap \P(F_2)$ (see \cref{fig:fc12}, left part, for two different examples).

\begin{figure}
	\includegraphics[scale=0.6]{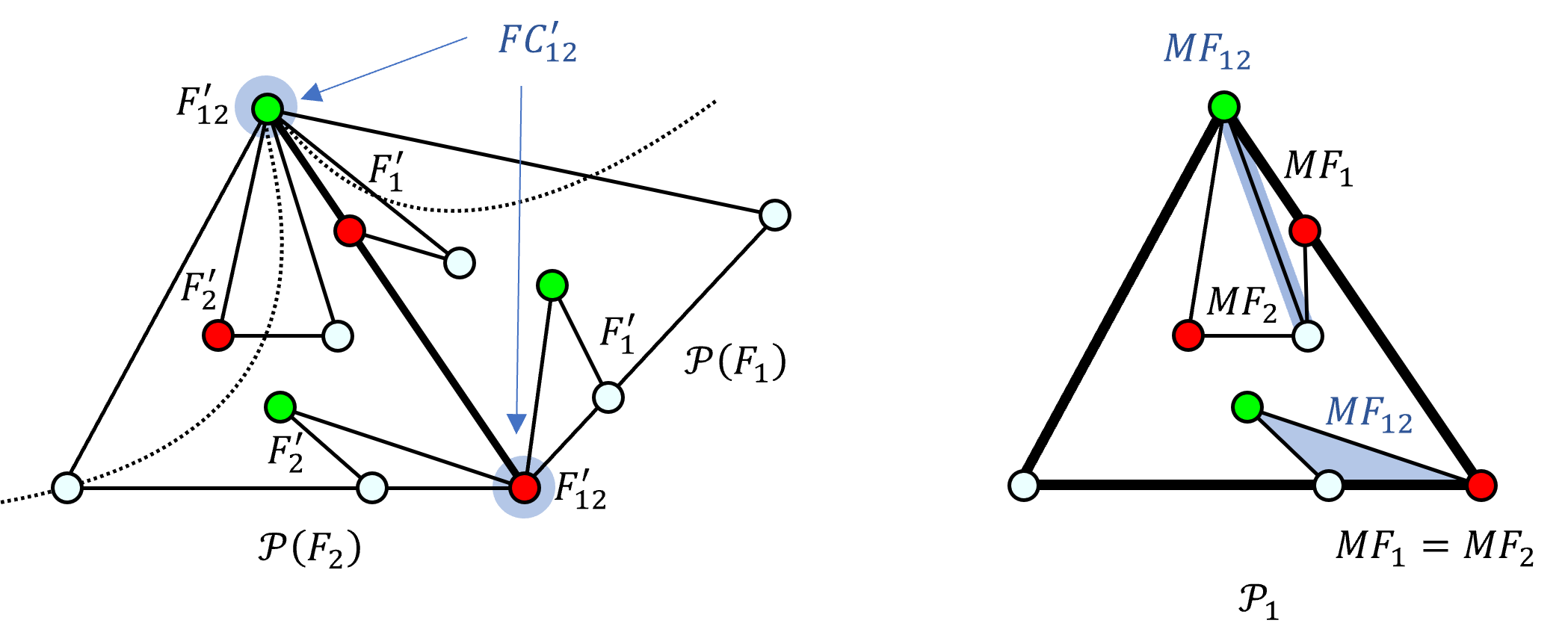}
	\centering
	\par\noindent\rule{0.6\textwidth}{0.5pt}
	\includegraphics[scale=0.6]{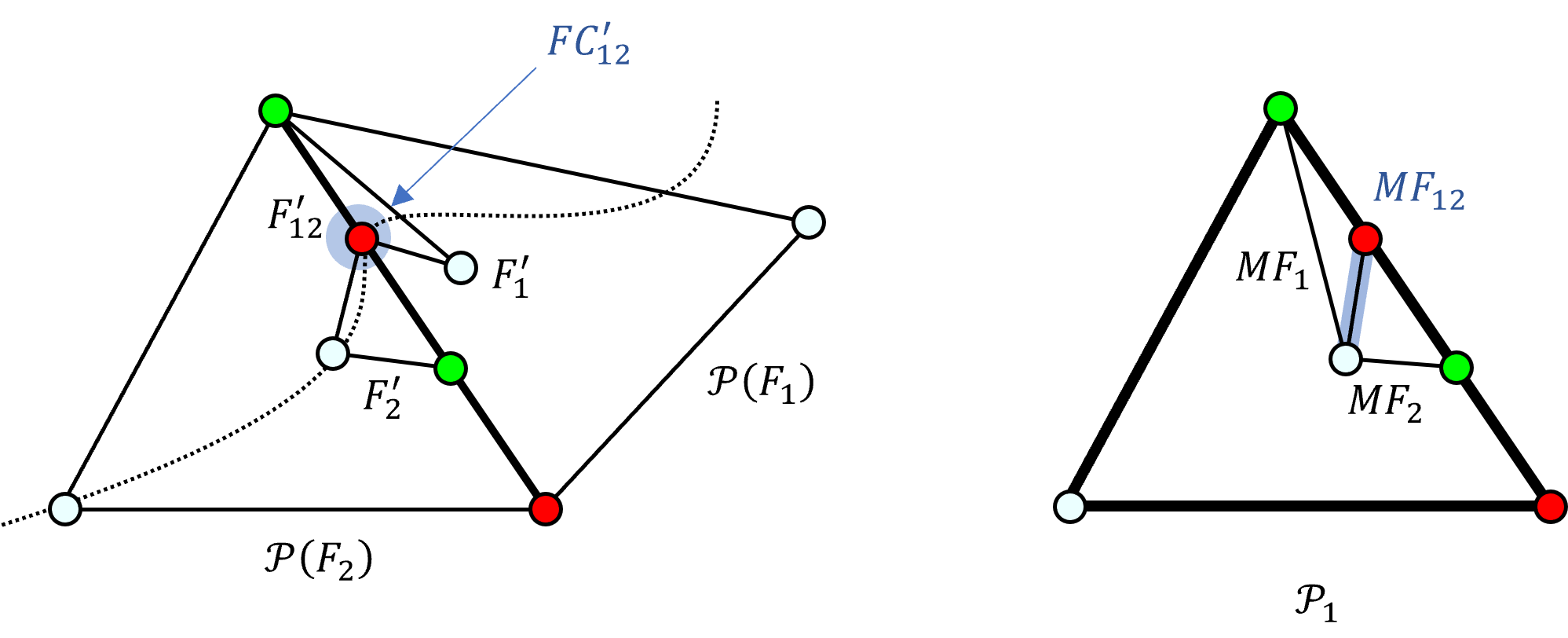}
	\centering
	\caption{The subcomplex $FC_{12}' \subseteq \P(F_1) \cap \P(F_2)$ in $\P_r$ and the corresponding subcomplex $MF_{12}$ in $\P_1$. Case~$(1)$ top and Case~$(2)$ bottom.}
	\label{fig:fc12}
\end{figure}

Any such $F_{12}'$ is not completely arbitrary, though: First of all,
since $FC_{12}' \subseteq \P(F_1) \cap \P(F_2)$ implies 
that its colors can only be drawn from $\chi(F_{12})$ due
to boundary consistency, we have
\begin{equation}
\chi(F_{12}') \subseteq \chi(FC_{12}') \subseteq \chi(F_{12}).\label{eq:nodeshrinking}
\end{equation}
Moreover, every pair of properly intersecting facets $F_1'$ and $F_2'$ is actually 
created by two unique \emph{matching border facets} $MF_1, MF_2 \in \P_1$: the adjacent
facets $F_1' \in \P(F_1)$ and $F_2' \in \P(F_2)$
are isomorphic to some
intersecting border facets $MF_1 \in \P_1$ and $MF_2 \in \P_1$, respectively, which
match at the boundary (see \cref{fig:fc12}, right part). 
This actually leaves only two possibilities for their intersection
$MF_{12}=MF_1 \cap MF_2 \neq \emptyset$:
\begin{enumerate}
\item[(1)] $MF_1$ and $MF_2$ are proper border facets with the same root component $R(MF_1)=R(MF_2) \in MF_{12}$ (possibly 
$MF_{12}\setminus R(MF_1)\neq \emptyset$, though). Two instances are shown in the top part of \cref{fig:fc12}.
\item[(2)] $MF_1$ and $MF_2$ are proper border facets with
  different root components, or improper border facets, with
  $R(MF_1)\cup R(MF_2) \not\subseteq MF_{12}$ (taken as a complex). An instance is shown in the bottom part of \cref{fig:fc12}.
\end{enumerate}
Note that these are all possibilities, since our single-rootedness assumption rules out
$R(MF_1)\cup R(MF_2) \subseteq MF_{12}$: After all, every
$v\in R(MF_1)\setminus R(MF_2) \neq\emptyset$  (w.l.o.g.)  would
need to have an outgoing path to every vertex in $R(MF_2)$, 
which is not allowed for the root component $R(MF_2)$
by \cref{def:root}. Keep in mind, for case (2) below, that
every vertex in $MF_{12}$ must have
an incoming path from every member of $R(MF_1)$ in $MF_1$
and from every member of $R(MF_2)$ in $MF_2$.

Now, given any pair of facets $F_1'$ and $F_2'$, we will consider conditions
ensuring the lifting/breaking of paths that run over their intersection 
$F_{12}'$. Not surprisingly, we will need to distinguish the two cases (1) and
(2) introduced above. To support the detailed description of 
the different situations that can happen here, we recall the path in 
$\C\subseteq \P_{r-1}$ that forms our starting point (\cref{fig:pr1}): It starts out from the proper border facet $\hF_x$ and leads to $\hF_1$, where it
enters the subcomplex $\S_1=\P^{r-2}(\tau_1)$. Within $\S_1$, the path continues to $F_1$.
The latter has a non-empty intersection $F_{12}$ with $F_2$, which belongs to the subcomplex $\S_2=\P^{r-2}(\tau_2)$.
The path continues within $\S_2$ and exits it at $\hF_2$, from where it finally leads to
the proper border facet $\hF_y$. 

\begin{enumerate}
\item[(0)] If $F_{12}'=\emptyset$, i.e., the
border facets $F_1'$ and $F_2'$ do not intersect at all,
there obviously cannot be any path in $\P_r$ that connects $\hF_1'$ and $\hF_2'$ via $F_{12}'$.

\item[(1)] If $F_{12}' \neq \emptyset$ is caused by proper 
border facets $F_1'$ and $F_2' \in \P_r$ with the same border root component
$R(F_1')=R(F_2')$, then both $F_1' \in \P_r$ and $MF_1 \in \P_1$ are isomorphic 
to some proper border facet $BF_1 \in \P_{r-1}$; analogously, $F_2' \in \P_r$ and $MF_2 \in \P_1$ 
are both isomorphic to some proper border facet $BF_2 \in \P_{r-1}$. This holds since
$\S_1'=\P^{r-1}(\tau_1)$ and $\S_2'=\P^{r-1}(\tau_2)$ are both isomorphic to $\P_{r-1}$.
Note carefully, though, that $BF_{12}=BF_1 \cap BF_2$ is isomorphic to $F_{12}'$ 
(but not necessarily to $MF_{12}$, as can be seen in the top of \cref{fig:fc12}),
so $\chi(R(BF_1))= \chi(R(BF_2)) 
\subseteq \chi(BF_{12}) \subseteq \chi(F_{12}') \subseteq \chi(F_{12})$ by \cref{eq:nodeshrinking}.
Note that \cref{def:bcrc} immediately implies $B(BF_1)=B(BF_2) \subseteq BF_{12}$
for the respective border components as well.

Depending on $BF_1$ and $BF_2$ (actually, depending on $BF_{12}$ and, ultimately, on $R(BF_1)= R(BF_2)$,
which we will say to \emph{protect} $F_{12}$), all paths running over $R(F_1')=R(F_2')$ will either~(a) 
be lifted  or else (b) break:
\begin{enumerate}
\item[(a)] We say that $R(BF_1)= R(BF_2)$ \emph{successfully} protects $F_{12}$ (see \cref{fig:lifting}) if $BF_{12}\in \C$, 
  i.e., if both $BF_1$ and $BF_2$ are within the same connected component $\C$
  as $F_1$ and $F_2$ (which also implies that 
there are paths in $\C$ connecting $\hF_1$ to $BF_1$ and $\hF_2$ to $BF_2$). In this
case, there is a lifted path in $\P_r$ connecting some border facets
$\hF_1'\in\P(\hF_1)$ and $\hF_2'\in \P(\hF_2)$ via $BF_{12}$, carried by the one in $\P_{r-1}$
that connected the proper border facets $\hF_1$ and $\hF_2$ via $F_{12}$: It exists, since
both $\S_1'=\P^{r-1}(\tau_1)$ and $\S_2'=\P^{r-1}(\tau_2)$ are isomorphic to $\P_{r-1}$.
By applying $\P$ to all remaining facets on the path that connected 
$\hF_x$ and $\hF_y$ in $\P_{r-1}$ as well, a path in $\P_r$ may be created that 
connects some incompatible proper border 
facets $\hF_x', \hF_y' \in \P_{r}$; of course, this requires a successful lifting everywhere
along the original path, not just at the intersection between $F_1$ and $F_2$. 
\item[(b)] We say that $R(BF_1)= R(BF_2)$ \emph{unsuccessfully} protects $F_{12}$ if $BF_{12} \not\in \C$. In this case, 
there cannot be a
path connecting the border facets $\hF_1'$ and $\hF_2'$ in $\P_r$ running via
$BF_{12}$, i.e., 
the connecting path in $\P_{r-1}$ cannot be lifted to $\P_r$ and thus breaks! 
\end{enumerate}

\begin{figure}
	\centering
	\includegraphics[scale=0.6]{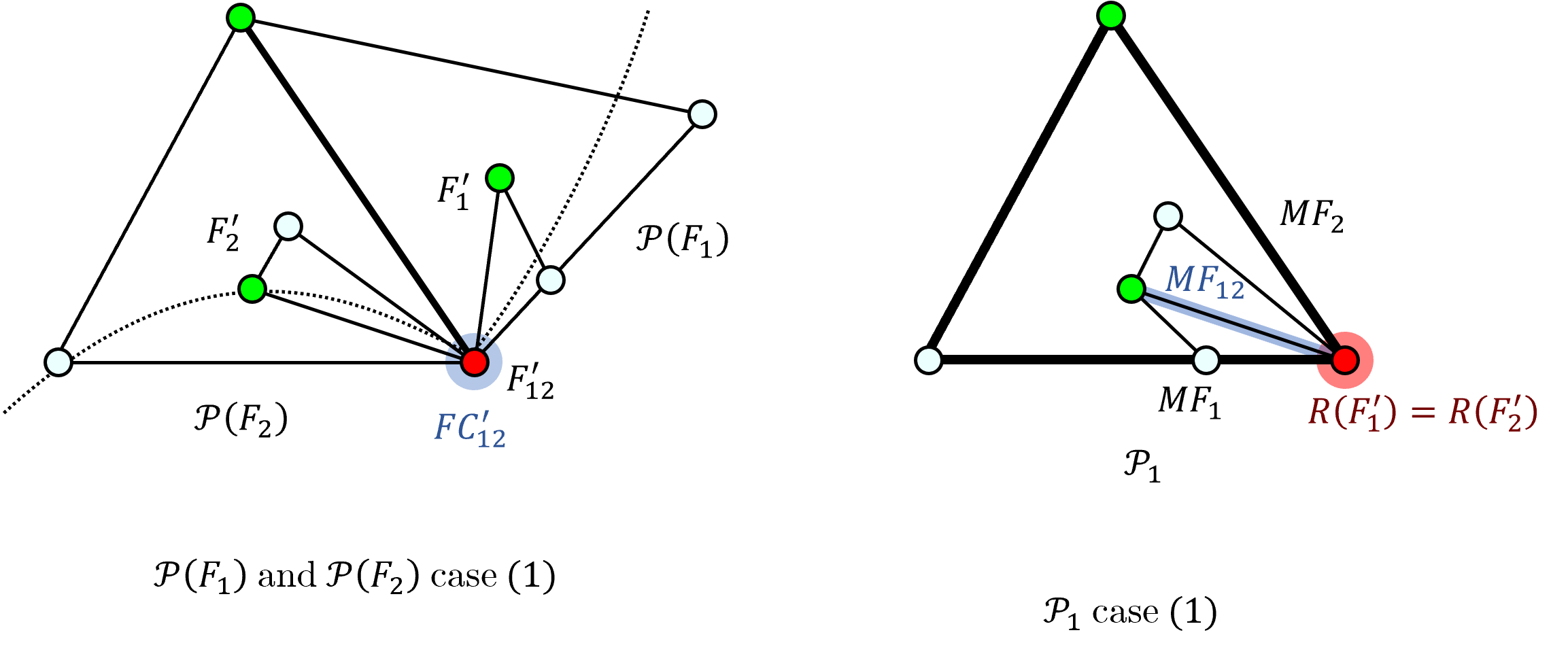}
	\linebreak
	$\P(F_1)$ and $\P(F_2)$ in $\P_r$, and the corresponding subcomplexes $MF_{1}$ and $MF_{2}$ in $\P_1$ 
	\linebreak
	\linebreak
	\includegraphics[scale=0.6]{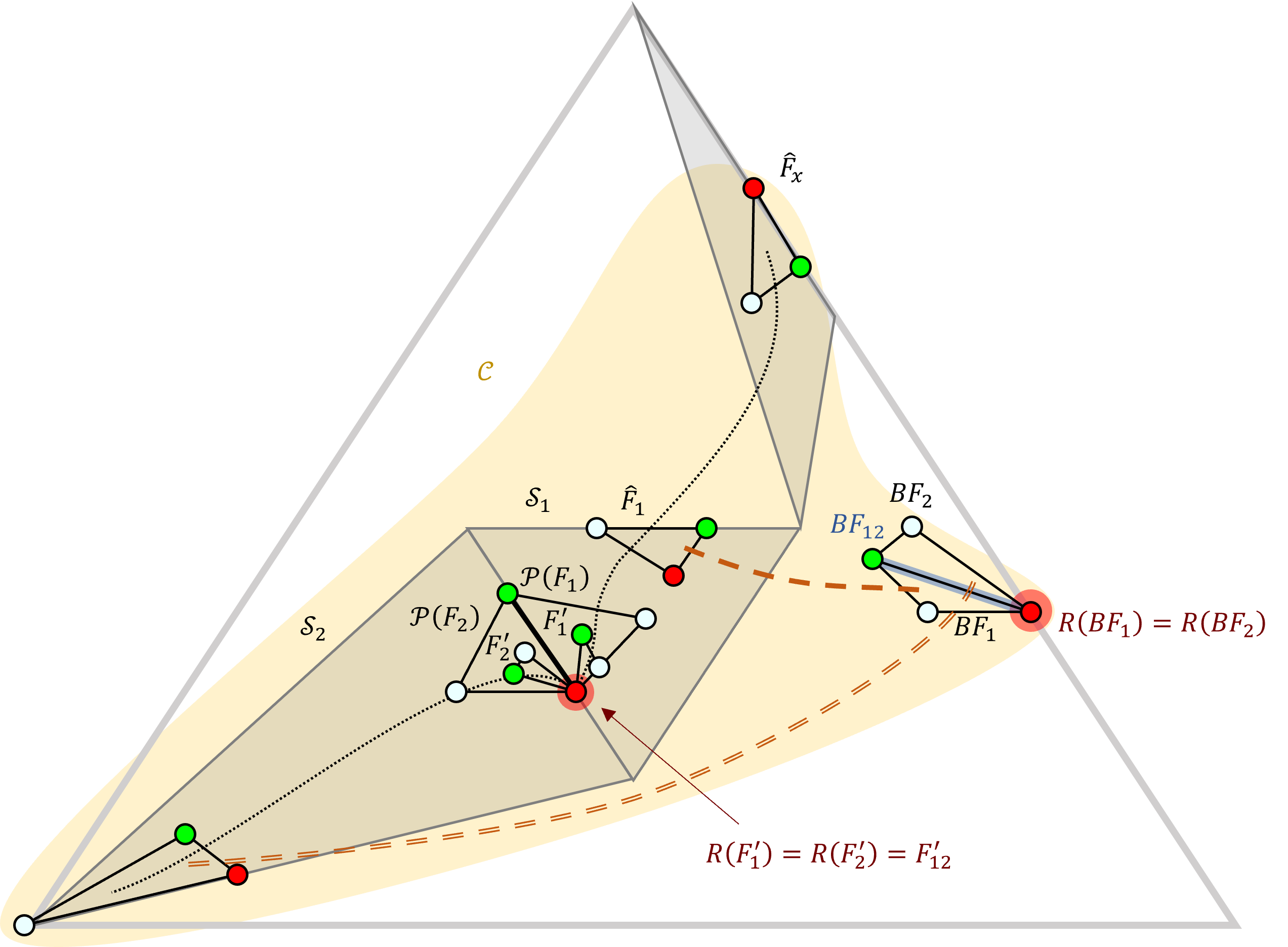}	
	\linebreak  $\P_{r-1}$ \linebreak\linebreak
	\includegraphics[scale=0.6]{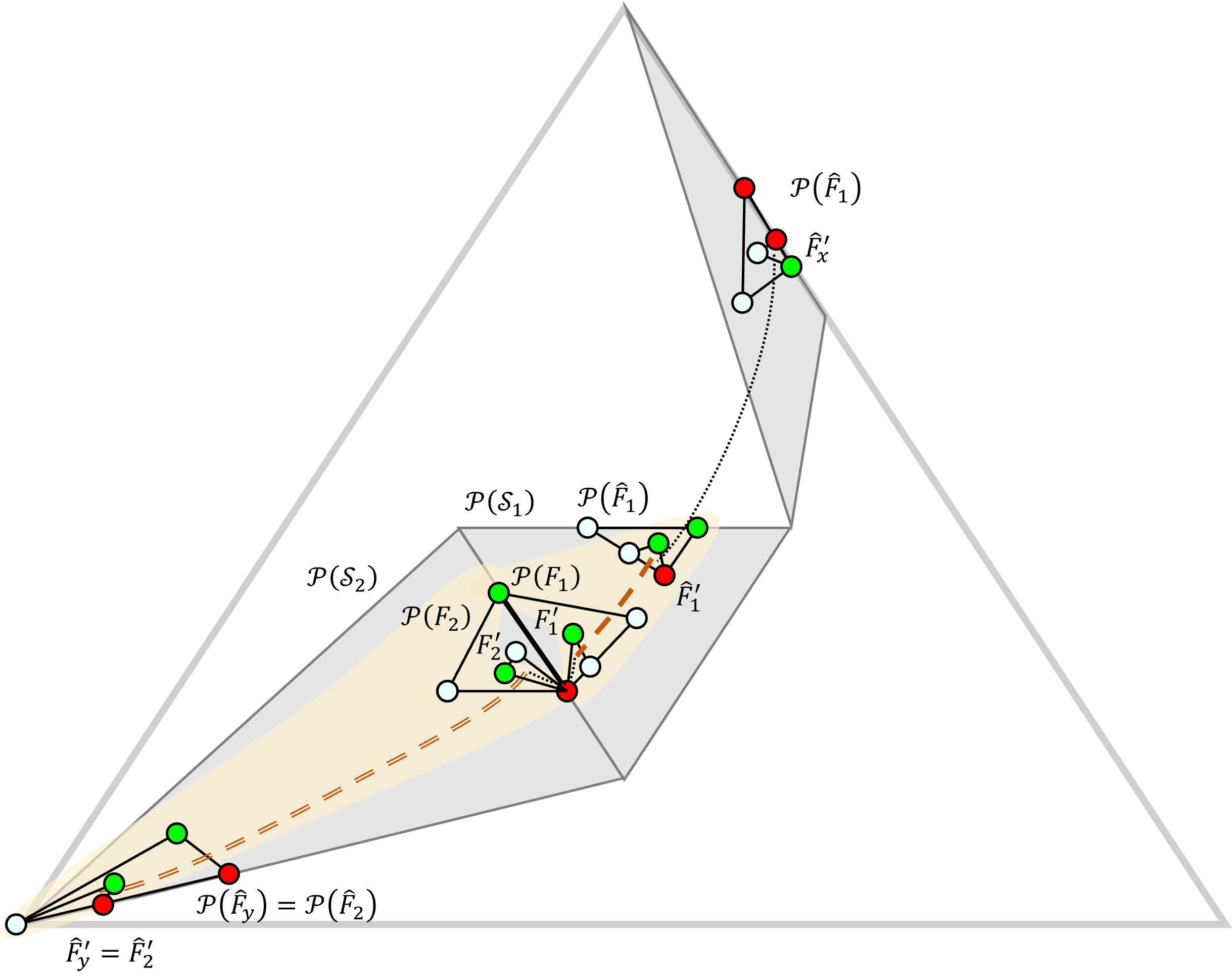}
	\linebreak  $\P_{r}$ with the images of $\C$ in $\S_1$ and $\S_2$ marked\linebreak
	\caption{Lifting --- Case 1(a)}
	\label{fig:lifting}
\end{figure}

\item[(2)] If $F_{12}' \neq \emptyset$ is not caused by proper 
  border facets $F_1' \in \P_r$ and $F_2' \in \P_r$ with
  common border root components
  $R(F_1')=R(F_2')$, we know from (2) above that
  $R(MF_1)\cup R(MF_2) \not\subseteq MF_{12}$, i.e., at least one of
  the root components, say, $R(F_1')$, has a vertex $v_1'\in F_1'$ outside
  $F_{12}'$. Clearly, the corresponding vertex $v_2'\in F_2'$ with
  $\chi(v_1')=\chi(v_2')$ is also outside $F_{12}'$ and hence
  different from $v_1'$. Since there is a path from $v_1'$ to
  every vertex in $F_1'$, at least one member in the intersection
  $F_{12}$ will be gone in $F_{12}'$, so
\begin{equation}
|V(F_{12}')| < |V(F_{12})|. \label{eq:propershrinking2}
\end{equation}
\end{enumerate}

\noindent
This completes the (exhaustive) list of cases that need to be considered w.r.t.
a single pair of facets $F_1'$ and $F_2'$.
Clearly, 
in order for a lifted path connecting $\hF_x'$ and $\hF_y'$ in $\P_r$ to break, 
it suffices that it breaks for just one pair of adjacent facets.
On the other hand, for a given pair $F_1, F_2 \in \P_{r-1}$, many paths are potentially
created simultaneously in $\P_r$, each corresponding to a possible selection 
of $F_1'$ and $F_2'$ and the particular intersection facet $F_{12}'$, which all need to break eventually.
Moreover, there are different paths in $\P_{r-1}$ connecting $\hF_x$ and $\hF_y$ via
different pairs $F_1, F_2$ that need to be considered. 
In \cref{sec:terminationtime}, we will show that there is even another
subtle complication caused by case (1.b),
the case where there is \emph{no} lifted path in $\P_r$: It will
turn out that ``bypassing'' may create a \emph{new} path connecting some incompatible proper
border facets in $\P_r$ when the path in
$\P_{r-1}$ breaks; see \cref{fig:multipath} for an example.

Finally, for consensus to be solvable, \emph{no} connected component $\C$ containing all border facets 
$\hF_1,\dots,\hF_k \in \C$ with incompatible 
border components $\hB_1,\dots,\hB_k$ may exist. In other words,
there must be some $r$ such that \emph{none} of the connected components
of $\P_r$ contains facets 
with incompatible border components. If this is ensured, the processes in any 
facet $F\in C\subseteq \P_r$ can eventually decide on the initial value of a deterministically 
chosen process in $\bigcap_{F \in \C, B(F)\neq\emptyset} \B(F) \neq \emptyset$. 
Note carefully, however,
that this also requires that all connections between incompatible borders that are 
caused by facet borders different from the border component $B(F)$ have disappeared.
Since this is solely a matter of case (2), \cref{eq:propershrinking2} reveals that 
this must happen after at most $n-1$ additional rounds.

\subsection{Characterizing consensus solvability via connected components}
\label{sec:charconnectedcomp}

It is enlightening to view cases (0)--(2) introduced before
w.r.t. the effect that they cause on the connected component $\C \subseteq
\P_{r-1}$ that connects incompatible border facets: 
Reconsider the two adjacent facets $F_1, F_2 \in \C$ with intersection
$F_{12}$, and assume, for clarity of the exposition, that $\C$ would fall
apart if the path running over $F_{12}$ would break. 
We will now discuss
what happens w.r.t.\ the connected component(s) in $\P_r$ when going to
$F_1' \in \P(F_1)$ and $F_2'\in \P(F_2)$, under the assumption that
$F_{12}'$ is the only intersecting facet in $\P_r$, according to our three cases: 
\begin{enumerate}
\item[(0)] If $F_{12}'=\emptyset$, then $\P_r$ would contain two separate
  connected components $\C_1'$ and $\C_2'$ with $F_1' \in \C_1'$ and
  $F_2' \in \C_2'$, i.e., the connected component(s) in $\P_r$
  resulting from $\C$ are separated by what is generated from $F_{12}$, namely,
  $F_{12}'=\emptyset$. A nice example is the (shaded) green center
node in $\P_2^{RAS}$ of the RAS message adversary shown in \cref{fig:RAScomplexes}.

\item[(1)] If $F_{12}'\neq\emptyset$ is caused by proper 
  border facets $F_1' \in \P_r$ and $F_2' \in \P_r$ (isomorphic
  to $BF_1$ resp.\ $BF_2$) with the same
  border root component
  $R(F_1')=R(F_2') = F_{12}'$ that successfully
  protects $F_{12}$, we have our two subcases:
  \begin{enumerate}
\item[(a)] If $BF_{12}\in \C$, then $\P_r$ would contain a single connected 
component $\C'$ (resulting from $\C$) that connects incompatible border facets.
An instance of such a successful protection 
can be found in \cref{fig:iRAScomplexes} of $\P_2^{iRAS}$ for the iRAS
message adversary, where consensus is impossible. Note that
just one communication graph has been added to RAS here, namely, 
the additional triangle that connects the bottom left white vertex to the central triangle 
in the 1-round uninterpreted complex $\P_1^{iRAS}$ in the top left part of \cref{fig:iRAScomplexes}. 
Consider the left border of the dash-dotted
central triangle, for example, where two adjacent facets intersect
in the common root $r \leftrightarrow g$. It results from the fact
that the border root component $r \leftrightarrow g$ of the proper
border facet on the right outermost border of the protocol
complex $\P_1^{iRAS}$ protects the intersection $r \leftrightarrow g$
of the central facet and the facet left to it in $\P_1^{iRAS}$.

\begin{figure}
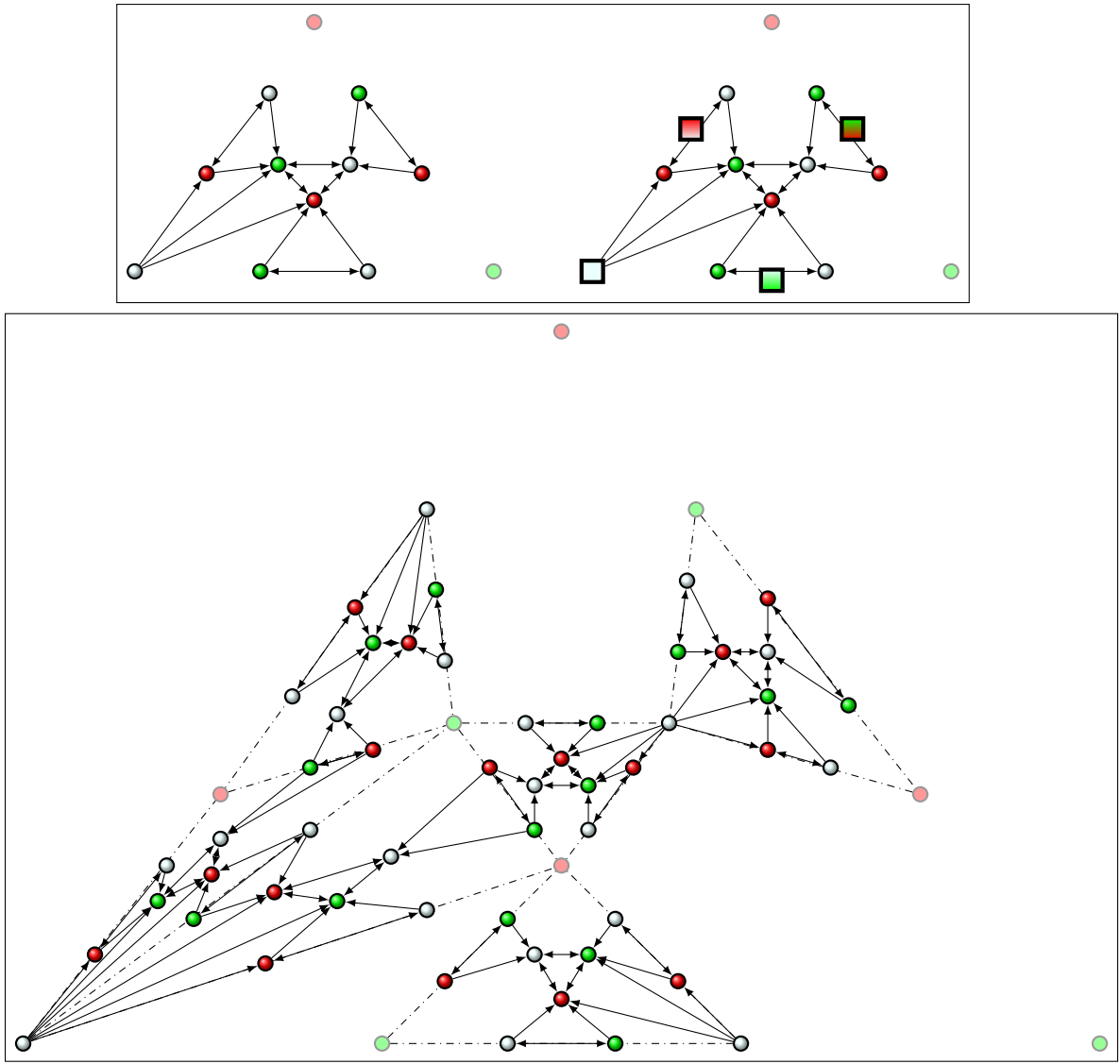

\ctikzfig{Figures/iras1}
\ctikzfig{Figures/iras2}
\caption{Protocol complex for one round ($\P=\P_1^{iRAS}$, top) and two rounds ($\P_2^{iRAS}$, bottom) of the iRAS message adversary. The top right figure also shows the border root components of $\P$.}
\label{fig:iRAScomplexes}
\end{figure}

\item[(b)] If $BF_{12}\not\in \C$, then $\P_r$ would contain two
  connected components $\C_1'$ and $\C_2'$ with $F_1' \in \C_1'$ and
  $F_2' \in \C_2'$. Unlike in case (0), however, they are separated
by a third connected component $\C_{12}'$ that contains $F_1'$ and $F_2'$. It 
can be viewed as an ``island'' that develops around $F_{12}'$. A nice
example of such an unsuccessful protection is the connected component containing the red process
in the center of \cref{fig:2Ccomplexes} for $\P_2^{2C}$ for the two-chain
message adversary, which now separates the single connected component
containing this process in $\P_1^{2C}$.

We note that the bypassing effect already mentioned (and discussed in detail 
in \cref{sec:terminationtime})
can also be easily explained via this view: It could happen that the
``island'' $C_{12}'$ is such that it connects some \emph{other}
incompatible proper border facets in $\P_r$ (see \cref{fig:multipath} for an 
example). So whereas it successfully
separates $\C_1'$ and $\C_2'$, it creates a new path that prohibits
the termination of consensus in round $r$.
\end{enumerate}

\item[(2)] If $F_{12}' \neq \emptyset$ is not caused by a proper 
  border facet $F_1' \in \P_r$ and $F_2' \in \P_r$ with common
  border root component
$R(F_1')=R(F_2') = F_{12}'$, $\P_r$ would 
contain a single connected component $\C'$ (resulting from $\C$) that still
connects incompatible border facets. 
\end{enumerate}

\subsection{Examples}
\label{sec:examples}

An example of an unsuccessful protection (a breaking path, case (1.b)) 
can be found in the 1-round uninterpreted complex $\P_1^{RAS}$ for
the RAS message adversary in the top-right part of \cref{fig:RAScomplexes}, where the facets $\hF_x$ and $\hF_y$ 
containing the border root components $\hR_x$ (the single white vertex in the bottom-left corner) and $\hR_y$ (the bidirectionally connected red and green vertices on the right border) are connected by a path that 
runs over the bottom leftmost triangle $F_1=\hF_x$ and the central triangle $F_2$, in a joint connected
component $\C \subseteq \P_1^{RAS}$. Note that $F_1$ and 
$F_2$ intersect in the single green central vertex $F_{12}=F_1 \cap F_2 = \{g\}$, and that there is no facet with a
border (root) component consisting only of the green vertex in $\P_1$ and hence in $\C$. Consequently,
it follows that $\P_2^{RAS}$ cannot contain a corresponding path connecting $\hF_x'$ with border root
component $\hR_x'$ (the single white vertex in the bottom-left corner in the
bottom part of \cref{fig:RAScomplexes}) and $\hR_y'$ (the bidirectionally connected red and green vertices 
on the right outer border) running over the faded central vertex $F_{12}'= \{g\}$, as is confirmed by our figure.

For an example of a successful protection (a non-breaking path, case (1.a)), consider the path connecting the border facets $\hF_x$ and $\hF_z$ 
containing the border root components $\hR_x$ (the single white vertex in the bottom-left corner) and $\hR_z$ (the bidirectionally connected red and white vertices on the left border) in the top-right part of \cref{fig:RAScomplexes}.
This path only consists of the bottom leftmost triangle $F_1=\hF_x$ and the triangle $F_2=\hF_z$ in a joint connected
component $\C\subseteq \P_1^{RAS}$. Note that $F_1$ and 
$F_2$ intersect in a red-green edge $F_{12}=F_1 \cap F_2=\{r \to g\}$ here, and that there is the border facet $\hF_y \in \C$ with a
border root component $\hR_y=\{g \leftrightarrow r\}$ on the rightmost outer border. 
According to our considerations above,
$\P_2^{RAS}$ contains a corresponding path connecting $\hF_x'$ with border root
component $\hR_x'$ (the single white vertex in the bottom-left corner in the
bottom part of \cref{fig:RAScomplexes}) and the border facet $\hF_z'$ with border root component 
$\hR_z'$ (the bidirectionally connected red and white vertices 
on the leftmost outer border) running via $F_{12}'=\{g \leftrightarrow r\}$, as is confirmed by our figure.

To further illustrate the issue of successful/unsuccessful protection, consider  
the modified RAS message adversary iRAS depicted in \cref{fig:iRAScomplexes}, where consensus is impossible. 
The border facets $\hF_w$ (the additional triangle) resp.\ $\hF_y$ 
containing the border root component $\hR_w$ (the single white vertex in the bottom-left corner) resp.~$\hR_y$ (the bidirectionally connected red and green vertices on the right border) are connected by a path that 
runs over the central bidirectional red-green edge $F_{12}=F_1 \cap F_2=\{g \leftrightarrow r\}$ in $\C$ here. In sharp contrast to RAS, 
the border facet $\hF_y$ with the border root component $\hR_y=\{g \leftrightarrow r\}$ on the right outer border is now also in $\C$, however. Consequently,
$\P_2^{RAS}$ contains a corresponding path connecting $\hF_w'$ with border root
component $\hR_w'$ (the single white vertex in the bottom-left corner in the
bottom part of \cref{fig:iRAScomplexes}) and $\hF_y'$ with border root component $\hR_y'$ 
(the bidirectionally connected red and green vertices 
on the right outer border) running via $F_{12}'=\{g \leftrightarrow r\}$, as is confirmed by our figure. 
Note that this situation recurs also in all further rounds, making consensus impossible.

To illustrate the issue of delayed path breaking (case (2)), 
consider another message adversary, called the \emph{2-chain message adversary} (2C), shown for $n=4$
processes in \cref{fig:2Ccomplexes} (top part). It consists of three graphs, a chain $G_1 = y \to g \to w \to r$, another chain $G_2 = g \to y \to w \to r$, and a star $G_3 = r \to \{y,w,g\}$. In $\P_1$, the facets $F_1$ and $F_2$, corresponding to $G_1$ and $G_2$, respectively, are connected by a path
running over the intersection $F_{12}=\{r\}$ in a joint connected component $\C$. 
There is also a border root component $R=\{r\}$ in 
the facet $F_3$ resulting from $G_3$, which, however, lies in a \emph{different} connected component $\C'\neq \C$
in $\P$. According to our considerations (case (1.b), the path (potentially) connecting
$F_1'$ (the border facet representing $G_1$ both in round 1 and 2) and $F_2'$ (the border facet representing $G_2$ both in round 1 and 2)
via $F_{12}'=\{r\}$ in $\P_2^{2C}$ breaks: As is apparent from the bottom part of \cref{fig:2Ccomplexes}, there
is no single red vertex shared by these two facets. 

If one adds another process $p$ (pink) to 2C for $n=5$, denoted by the message adversary
2C+, such that 
$G_1=y \to g \to w \to p \to r$, 
$G_2=g \to y \to w \to p \to r$, and $G_3 = r \to \{y,w,g,p\}$, then $F_{12}=\{p \to r\}$ is in $\P_1^{2C+}$.
Now there is a path in $\P_2^{2C+}$ connecting
$F_1'$ (the border facet representing $G_1$ both in round 1 and 2) and $F_2'$ (the border facet representing $G_2$ both in round 1 and 2)
running via $F_{12}'=F_1' \cap F_2' = \{r\}$: Whereas the pink vertex has learned in round 2 
where it belongs to, i.e., either $F_1'$ and $F_2'$, from the respective root component, 
this is not (yet) the case for the red vertex. However, whereas the corresponding path
did not break in $\P_2^{2C+}$, it will finally break in $\P_3^{2C}$ since the red vertex will
also learn where it belongs to.

\begin{figure}
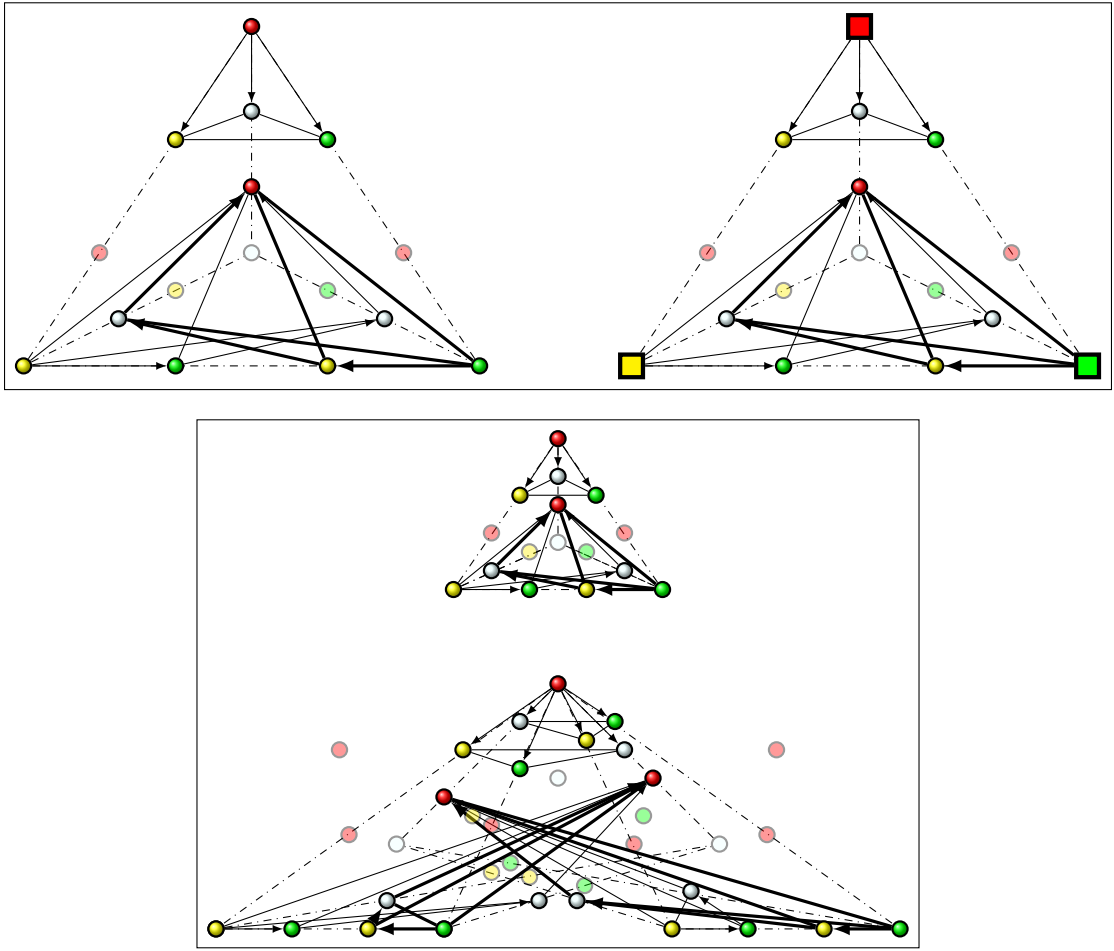

\ctikzfig{Figures/twochainexample}
\ctikzfig{Figures/twochainexample2}
\caption{Protocol complex for one round ($\P=\P_1^{2C}$, top) and two rounds ($\P_2^{2C}$, bottom) of the two-chain message adversary for $n=4$ processes. 
The top right figure also shows the border root components of $\P$.}
\label{fig:2Ccomplexes}
\end{figure}

\subsection{A decision procedure for consensus solvability}
\label{sec:decisionprocedures}

Revisiting the different cases (0)--(2) 
that can occur w.r.t.
lifting/breaking a path connecting incompatible border facets
in $\P_{r-1}$ to $\P_r$, it is apparent that the only case that
might lead to a path that never breaks, i.e., in no round
$r\geq 1$, is case (1.a): In case (0) and (1.b), there cannot
be a lifted path running via $F_{12}'$ in $\P_r$, i.e., the path 
in $\P_{r-1}$ breaks immediately. In case (2),
it follows from \cref{eq:propershrinking2}
that this type of lifting could re-occur in at most $n-2$ consecutive rounds
after a path running over $F_{12}$ is lifted to
a path running over $F_{12}'$ in $\P_r$ for the first time. Since
these are all possibilities, after the ``exhaustion'' of case (2),
$F_{12}'=\emptyset$ and 
hence case (0) necessarily applies.

In order to decide whether consensus is solvable
for a given message adversary $\MA$ at all, it hence suffices to keep track of case
(1.a) over rounds $1, 2, \dots$. If one finds that case (1.a) does not
occur for any path in $\P_{r-1}$ for some $r$, there is no need for
iterating further. On the other hand, if one finds that
case (1.a) re-occurs for some path forever, consensus is
impossible. Since the facets $BF_1$ and $BF_2$, where the common root component
$R(BF_1)=R(BF_2)$ successfully protects $F_{12}$ in case (1.a),
leads to $\chi(F_{12}') \subseteq \chi(F_{12})$ according to
\cref{eq:nodeshrinking}, the infinite re-occurence
of case (1.a) for some path implies that there is some round
$r_0$ such that $\chi(F_{12}') = \chi(F_{12}) = I \subset \Pi$ 
for all $r\geq r_0$. If this holds true, then case (1.a) must also re-occur perpetually
in the lifted paths obtained by using the same $BF_1$ and $BF_2$
with $\chi(BF_{12})=I$ in all rounds $r\geq r_0$.

For keeping track of possibly infinite re-occurrences of case (1.a), it is hence sufficient to determine,
for every pair of facets $MF_1\in \P_1$ and $MF_2 \in \P_1$, $MF_2\neq MF_1$,
intersecting in $MF_{12}\neq\emptyset$, 
the set of proper border facets $MF^1, \dots, MF^\ell \in \P_1$ with border 
root components $R^j=R(MF^j)$ satisfying 
$\chi(R^j) \subseteq \chi(MF_{12})$ for all $1 \leq j \leq \ell$. 
Clearly, every choice of $MF^j$, $MF^k$ is a possible candidate for the isomorphic
re-occurring protecting facets $BF_1=BF^j$ and $BF_2=BF^k$ for case (1.a) in 
some $\P_{r-1}$, provided (i) $R(MF^j)=R(MF^k)$ and (ii) both $BF^j$ and $BF^k$
are in the connected component $\C \subseteq \P_{r-1}$ containing $F_1$ and $F_2$. 
If (ii) does not hold, one can safely drop $MF^j$, $MF^k$ from the set of 
candidates for infinitely re-occurring protecting facets in all subsequent rounds.

This can be operationalized in an elegant and efficient decision procedure
by using an appropriately labeled and weighted version 
of the facets' \emph{nerve graph} $\N$ of the 1-round uninterpreted complex $\P_1$.
It is a topological version of the combinatorial decision procedure given in
\cite[Alg.~1]{WPRSS23:ITCS} that works as follows:
Every facet in $\P_1$ is a node $F$ in $\N$ and labeled by $w(F)=R(F)$, 
its root component in $\P_1$. Two nodes $F^1,F^2$ in $\N$ 
are joined by an (undirected) edge $(F^1,F^2)$, if they intersect in a simplex
$\emptyset \neq F^{12}=F^1\cap F^2$ in $\P_1$. The edge is labeled by $w((F^1,F^2)) =  \{R^1,\dots,R^\ell\}$ (possibly empty), which is the maximal set of
(necessarily: border) root components that satisfy the property
$\chi(F^{12})\supseteq \chi(R^i)$.
Recall that the member sets of different border root components may satisfy 
$\chi(R^i)\cap \chi(R^j) \neq \emptyset$ 
and even $\chi(R^i)=\chi(R^j)$, albeit $R^i \cap R^j \neq \emptyset$ when taken as faces is impossible.

The procedure for deciding on consensus solvability proceeds in iterations,
starting from $\N_0=\N$,
and defining $\N_{i+1}$ from $\N_{i}$ as follows.
Let
$V(\N_{i+1}):=V(\N_{i})$ with the same node labels $w(F)$, 
initialize $E(\N_{i+1})$ to be the empty set,
and add to it each edge $(F^1,F^2)\in E(\N_i)$ with a label $w_{i+1}((F^1,F^2))$ defined next, but only if this label is not empty.
For a potential edge $(F^1,F^2)\in E(\N_i)$, 
set $R \in w_{i+1}((F^1,F^2))$
if the (unique) connected component $C_{i}^j$ of $\N_i$ with $(F^1,F^2) \in E(\C_{i}^j)$ 
contains some $F'\in V(\C_{i}^j)$ with $w(F')=R \in w_{i}((F^1,F^2))$.
The construction stops when either (i) none of the connected components of $\N_i$ contains 
nodes representing facets with incompatible root components (consensus is solvable), or else (ii) if
$\N_{i+1} = \N_{i}$ but there is at least one connected component containing
nodes representing facets with incompatible root components (consensus is impossible).

For example, \cref{fig:RASn} shows the labeling of the facets with their root components for the RAS message
adversary, where consensus can be solved. The sequence of nerve graphs $\N$, $\N_0$ and $\N_1$ is illustrated in 
\cref{fig:RASnerve}. On the other hand, \cref{fig:iRASn} and \cref{fig:iRASnerve} show the same for the iRAS message
adversary, where consensus cannot be solved.

\begin{figure}
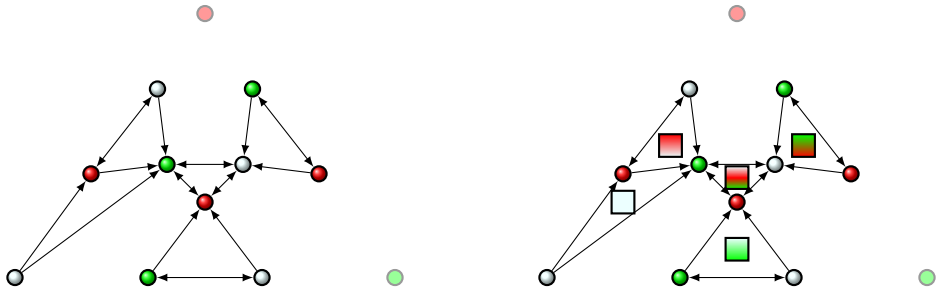

\ctikzfig{Figures/ras1n}
\caption{Results of labeling the faces of the 1-layer protocol complex 
$\P_1$ of the RAS message adversary (left) by their root component (right).}
\label{fig:RASn}
\end{figure}

\begin{figure}
\ctikzfig{Figures/ras1nerve}
\caption{Construction of the initial nerve graph $\N_0$ of the 1-layer protocol 
complex $\P_1$ of the RAS message adversary: After replacing the facets by their
corresponding nodes (colored by their root component) and labeling all the edges (left), nerve graph $\N_0$ 
obtained by removing edges without a protecting root component (middle),
nerve graph $\N_1$ (right). Note that $\N_1$ already reveals that 
consensus is solvable.}
\label{fig:RASnerve}
\end{figure}

\begin{figure}
\ctikzfig{Figures/iras1n}
\caption{Results of labeling the faces of the 1-layer protocol complex 
$\P_1$ of the iRAS message adversary (left) by their root component (right).}
\label{fig:iRASn}
\end{figure}

\begin{figure}
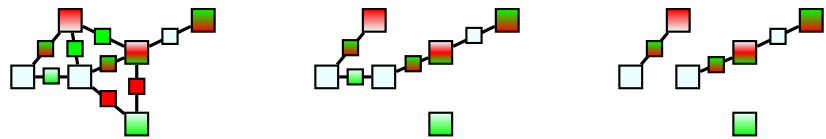

\ctikzfig{Figures/iras1nerve}
\caption{Construction of the initial nerve graph $\N_0$ of the 1-layer protocol 
complex $\P_1$ of the iRAS message adversary: After replacing all facets by their
corresponding nodes (colored by their root component) and labeling all the edges (left), nerve graph $\N_0$  
obtained by removing edges without a protecting root component (middle),
nerve graph $\N_1$ (right). Since $\N_2=\N_1$, which still contains a component that
connects incompatible root components, consensus is impossible.}
\label{fig:iRASnerve}
\end{figure}

Note that there is a small difference between the decision procedure \cite[Alg.~1]{WPRSS23:ITCS} and our
topological version: Whereas the latter uses sets of border root components $w_i((F^1,F^2)) = \{R^1,\dots,R^\ell\}$
as the label of an edge $(F^1,F^2)$, the size of which may decrease during the iterations, the former 
uses the fixed set of processes that cannot distinguish $F^1$ and $F^2$ in $\P_1$ as its label $\ell((F^1,F^2))$. 
The latter does not change during the iterations, and can in fact be written as $\ell((F^1,F^2)) = 
\chi(w_0((F^1,F^2))) = \chi(\{R^1,\dots,R^\ell\})$. Whereas these different labeling schemes are equivalent in
terms of correctly deciding consensus solvability/impossibility, ours might facilitate a more efficient data 
encoding and thus some advantages in computational complexity for certain message adversaries. On the other hand, we could of course also use the original
labeling of
\cite[Alg.~1]{WPRSS23:ITCS} in our decision procedure and detect successfully protecting border root components via proper inclusion of the member sets.


\section{Consensus Termination Time}
\label{sec:terminationtime}

In this section, we will shift our attention from the principal question of whether
consensus is solvable under a given message adversary $\MA$ to the question
of how long a distributed consensus algorithm may take to terminate.

Whereas it is immediately apparent that the number of iterations of the 
decision procedure in \cref{sec:decisionprocedures} is a lower bound for
the consensus termination time, their exact relation is not clear: 
Case (2) of our classification
for path lifting/breaking in \cref{sec:consensuschar} revealed
an instance where the actual breaking of a path may happen up to
$n-1$ rounds \emph{after} detecting that it will eventually break. 
An interesting question is whether there are other effects that may even increase this gap between iteration complexity of the decision procedure and consensus termination time. 
And indeed, \cite{WPRSS23:ITCS} provided an example that shows that
this gap may even be exponential in $n$. 
In \cref{sec:delayedpathbreaking},
we will provide an intuitive topological explanation of this gap, 
which is caused by the possibility of ``bypassing''. 
In \cref{sec:RRGdecisionproc}, we finally propose a decision procedure, 
which allows to answer the question whether 
distributed consensus is solvable in $k$ rounds under a given message adversary $\MA$.

\subsection{Delayed path breaking due to bypassing}
\label{sec:delayedpathbreaking}

We mentioned already in \cref{sec:charconsensus} that
in order for some incompatible border components to become disconnected, 
all paths connecting those must break. 
Consider the situation illustrated in \cref{fig:multipath}, 
for the case of a system of $n=5$ processes ($r, g, w, p, y$), and a message adversary that comprises only 5 graphs, 
according to the uninterpreted 1-round protocol complex $\P_1$ illustrated in the top part of the figure. 
There are two paths $P_{a}\in \C_a$ and $P_b\in C_b$ in $\P_1$
that connect the same incompatible border facets (touching upon the $\{y,g\}$ resp.\ upon the
$\{p,w\}$ border component carrier), 
lying in \emph{different} connected components $\C_a$ and $\C_b$. 
The left path $P_a$ consists of facets $F_1$ (with $R(F_1)=\{y\}$)
and $F_2$ (with $R(F_2)=\{w\}$), 
sharing the face $F_{12}=\{r\}$. 
The right path $P_b$ consists of $H_1$, $H_0$ and $H_2$ and involves the facet $H_0$ with border root component $R(H_0)=\{r\}$. 
According to case (1.b), both corresponding lifted paths in $\P_2$ break, 
since the shared faces $F_{12}$ between any two facets are (unsuccessfully) protected by the common root component of
proper border facets in $\P_1$ lying
in a \emph{different} connected component only.

\begin{figure}
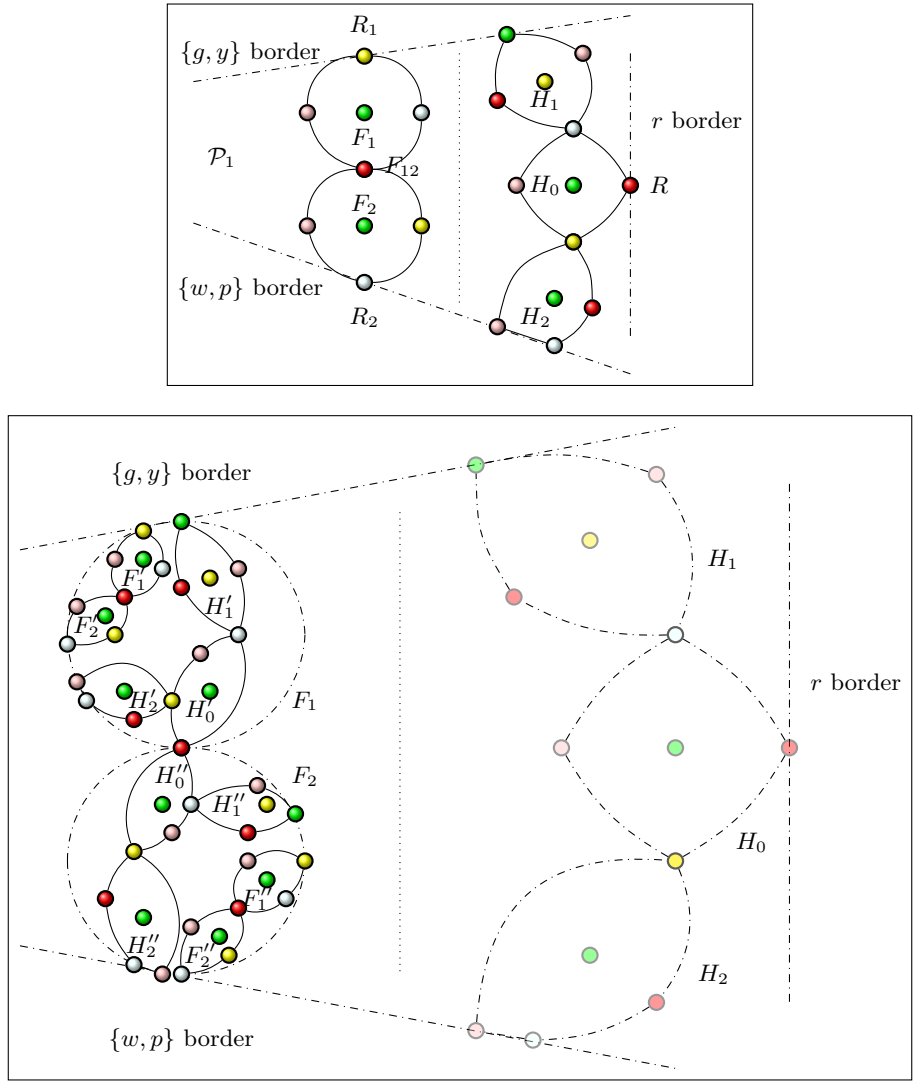

\ctikzfig{Figures/multipath}
\ctikzfig{Figures/multipath2}
\caption{Illustration of delayed path creation in the evolution of a protocol complex, for $n=5$. The top part shows $\P_1$, which consists of two paths $P_a$ and $P_b$
connecting incompatible border (root) components (on the $\{y,g\}$ resp.\ $\{p,w\}$ border), lying in different connected components $\C_a$ and $\C_b$. 
The left path $P_a$ consists of facets $F_1$ and $F_2$, sharing $F_{12}=\{r\}$. The
right path $P_b$ consists of $H_0$--$H_2$ and involves also a facet $H_0$ with border (root) component $R(H_0)=\{r\}$. Note that our restriction to $n=5$
implies that the pink vertices in $F_1$ and in $F_2$ are actually
the same, and so are the red vertices in $H_1$ and $H_2$.
The bottom part shows $\P_2$: Whereas the corresponding paths for both $P_a$ and $P_b$ break in round 2 according to case (1.b), it also happens that $P_b$ creates a lifted path in $\P_2$ (running within $\P(F_1)$ and 
$\P(F_2)$) that connects ``new'' incompatible proper border components there. This lifted path will break only in $\P_3$.}
\label{fig:multipath}
\end{figure}

Surprisingly, however, the $\{y,g\}$ and $\{p,w\}$ borders themselves are \emph{not} separated in
$\P_2$. Actually, it happens that the right path $P_b$ in $\C_b \subseteq \P_1$ gives rise
to a \emph{new} lifted path connecting facets with proper border components in the $\{y,g\}$ resp.\ 
$\{p,w\}$ borders in $\P_2$. This effect, called \emph{bypassing}, is illustrated in the bottom part of 
\cref{fig:multipath}: By applying $\P$ to both $F_1,F_2$, leading to $\P(F_1), \P(F_2) 
\subseteq \P_2$, one observes that $P_b$ now leads to a new lifted path connecting the incompatible border (root) components
$R(H_1')=\{g\}$ (in the facet $H_1'$ corresponding to $H_1$ in $\P(F_1)$) and $R(H_2'')=\{p,w\}$ (in the facet $H_2''$ 
corresponding to $H_2$ in $\P(F_2)$) via the intersection of $H_0' \cap H_0''=\{r\}$. 
In fact, the island created in
$\P_2$ around the latter, due to case (1.b), which consists of $H_1'$, $H_2'$, $H_0'$, $H_0''$, $H_1''$ and $H_2''$ and nicely
separates the connected components consisting of $F_1'$ and $F_2'$ from $F_1''$ and $F_2''$, 
is not an island, but rather
connects two other incompatible border root components, namely $\{g\} \in H_1'$ and 
$\{w\} \in H_2''$. 
Whereas it can be inferred already in $\P_1$ that this new lifted path in $\P_2$
will eventually break as well, consensus cannot be solved in just two rounds here. 

Even worse, for larger $n$, it is possible to iterate this construction: An additional path 
$P_c$ in a separate connected component $\C_c \subseteq \P_1$ could bypass both
the shared face $\{w\}$ between $H_1$ and $H_0$ and $\{y\}$ between $H_2$ and 
$H_0$ in $P_b$, in the same way as the shared face $\{r\}$ in $P_a$ is bypassed. More specifically, 
if these shared faces in $P_b$ are (unsuccessfully) protected by the border root components of proper border facets 
in $P_c \in \C_c$, which connect proper border facets touching the 
incompatible $\{g,y\}$ border to 
the $\{w\}$ border, to the $\{r\}$ border, to the $\{y\}$ border, and finally to the $\{w,p\}$ border, 
one gets a path connecting the $\{g,y\}$ and the $\{w,p\}$ borders in $\P_2$, carried by $P_b$ in $\P_1$, in exactly the same way as we got the path carried by $P_a$ described above. 
Note carefully that the border root component of the proper border facet touching the $\{r\}$
border in $P_c$ must be different from the one touching the $\{r\}$ border in $P_b$, since $\C_b \neq \C_c$.
Since $\P_3=\P(\P_2)$, this finally causes the creation of a new path in $\P_3$ that also connects the 
incompatible $\{g,y\}$ and $\{w,p\}$ borders, carried by $P_a$ in $\P_1$.

Whereas successive bypassing cannot go on forever, it stops
only if there are no ``new'' connected components in $\P_1$ that allow to bypass shared faces. 
Indeed, there are natural limits of the number of such bypassing connected components:
\begin{enumerate}
\item[(1)] The bypassing connected components $\C_1, \C_2, \dots$ in $\P_1$ must connect incompatible borders,
but must be disjoint. The root components of the facets that 
touch some specific border in different components must hence all be different (when taken as faces) 
as well. However,
the example worked out in \cite{WPRSS23:ITCS} demonstrates that there can be exponentially (in $n$)
many of those.

\item[(2)] A connected component $\C_{x+1}$ that contains a border root component $R$ that unsuccessfully
protects a 
shared face $F_{12}$ in the connected component $\C_x$ to accomplish bypassing 
must be such that it connects both the 
incompatible borders of $\C_x$ and some border containing $R$. Since each such $F_{12}$ must be
protected  by some proper border facet in $\C_{x+1}$, the length of the paths connecting
two particular borders must be strictly increasing.
\end{enumerate}

This ultimately provides a very intuitive ``geometric'' explanation of the quite
unexpected exponential blowup of the gap between the iteration complexity of the
decision procedure and the termination time of distributed consensus. In
particular, (1) explains the surprising fact that the \emph{number} of connected 
components in $\P_1$ plays a major role here.

\subsection{A decision procedure for $k$-round distributed consensus}
\label{sec:RRGdecisionproc}

Reviewing the decision procedure of \cref{sec:decisionprocedures} in the light of bypassing as described
in \cref{sec:delayedpathbreaking}, it is apparent that the nerve graph based
approach removes edges/labels \emph{eagerly}. Regarding decision time, this is of
course most advantageous: In the example of \cref{fig:multipath}, it would
terminate already after one iteration, telling that consensus is solvable.

There is a less eager alternative decision procedure, which builds a sequence 
of \emph{(border) root reachability graphs} $\RRG_i$, $i\geq 0$, that tell
which proper border facets are reachable from each other in $\P_{i+1}$.
First, it builds the initial root reachability graph
$\RRG_0$, the vertices of which (represented as square nodes in our illustrating
figures) are the border root components (which are the same as the border components
for all proper border facets) of the 1-round uninterpreted complex $\P_1$, 
see \cref{fig:RAScomplexes} 
(top right), and
where two such vertices are connected by an undirected edge if they are connected
via a path in $\P_1$ (irrespectively of the type of edges in $\P_1$), see \cref{fig:RASreachability} (top left). We obtain
$\RRG_1$ by replacing every facet $F$ in $\P_1$ by an instance of $\RRG_0$,
in such a way that the replacements of two facets $F_1$, $F_2$ that intersect in a 
simplex $F_{12}$ (case (1.a) in \cref{sec:charconsensus}) that is protected by the common root
$R(MF_1)=R(MF_2)$ of the proper border facets $MF_1$ and $MF_2$ in $\P_1$,
i.e., $\chi(R(MF_1))=\chi(R(MF_2)) \subseteq \chi(F_{12})$, share a node labeled 
$\chi(R(MF_1))=\chi(R(MF_2))$. Note that an actual
root component is represented by a fat square node in our figures, 
whereas the node representing an intersection is displayed by
a non-fat square node.

\cref{fig:RASreachability} shows $\RRG_0$ (top left part), obtained directly from the top-right part of
\cref{fig:RAScomplexes}, and $\RRG_1$ (top right part), which consists of several connected components. It is apparent, however, that it no longer connects incompatible border components. In particular, the bottom left border root component consisting of the white fat square node is
no longer connected by a path to the red-green fat square node on the right
side of the outer triangle in $\RRG_1$. That is, the connection between these
two border root components (present in $\RRG_0$) has disappeared!

This immediately gives us a recursive procedure for deciding consensus solvability: 
Rather than starting from the initial $\RRG_0=\RRG_0^{(0)}$, 
we start inductively from the previously constructed $\RRG_0^{(i)}$, $i\geq 0$,
and plug it into $\P_1$ exactly as before to construct $\RRG_1^{(i+1)}$. Note that,
for $i\geq 1$, $\RRG_1^{(i+1)}$ has at most the same number of edges than $\RRG_1^{(i)}$. 
This process can be 
repeated until $\RRG_1^{(m+1)}=\RRG_1^{(m)}$ for some $m\geq 0$. Consensus is possible if and only
if $\RRG_1^{(m)}$ contains no component that connects incompatible fat square nodes. In the example of \cref{fig:RASreachability}, already $\RRG_1^{(1)}$ does
not connect incompatible fat square nodes, so consensus is solvable under RAS.

This recursive RRG construction is equivalent to the the following iterative procedure, which
operates directly on the root reachability graphs: 
Starting out from $\RRG_{i}$, initially
$\RRG_1$, construct $\RRG_{i+1}$ by removing every edge incident to a node (= a non-fat 
square node) where the common border root component (= the fat square node $R(BF_1)=R(BF_2)$) 
of the protecting
matching border facets is not in the same connected component. The procedure stops if
either the resulting $\RRG_i$ contains no component that connects incompatible fat square nodes (in
which case consensus is solvable), or else if $\RRG_{i+1}=\RRG_{i}$ (in which case consensus
is not solvable if incompatible fat square nodes are still connected).

%

\begin{figure}
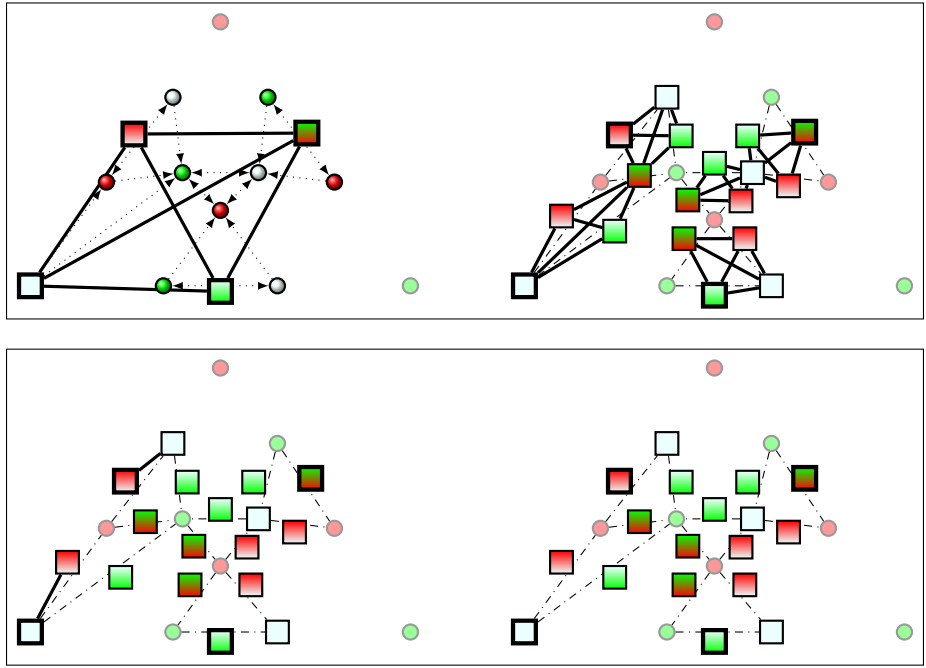

\ctikzfig{Figures/ras1reachability}
\ctikzfig{Figures/ras1reachability2}
\caption{Construction of the root components reachability graphs $\RRG_0$--$\RRG_3$ from the
1-layer protocol complex for RAS: Initial root reachability graph $\RRG_0$ (top left) and
$\RRG_1$ (top right). Since $\RRG_1$ already partitions into several connected
components that no longer contain incompatible border root components, one can already
decide here that consensus is solvable. For completeness, we also show 
$\RRG_2$ (bottom left) and $\RRG_3$ (bottom right), where all edges have finally been removed.}
\label{fig:RASreachability}
\end{figure}

\cref{fig:iRASreachability} and \cref{fig:iRASn} show the RRG construction
for the iRAS message adversary, where consensus is impossible, as
introduced in \cref{fig:iRAScomplexes}. The case of the 2-chain message
adversary 2C is illustrated in in \cref{fig:twochainexamplereachability}.

\begin{figure}
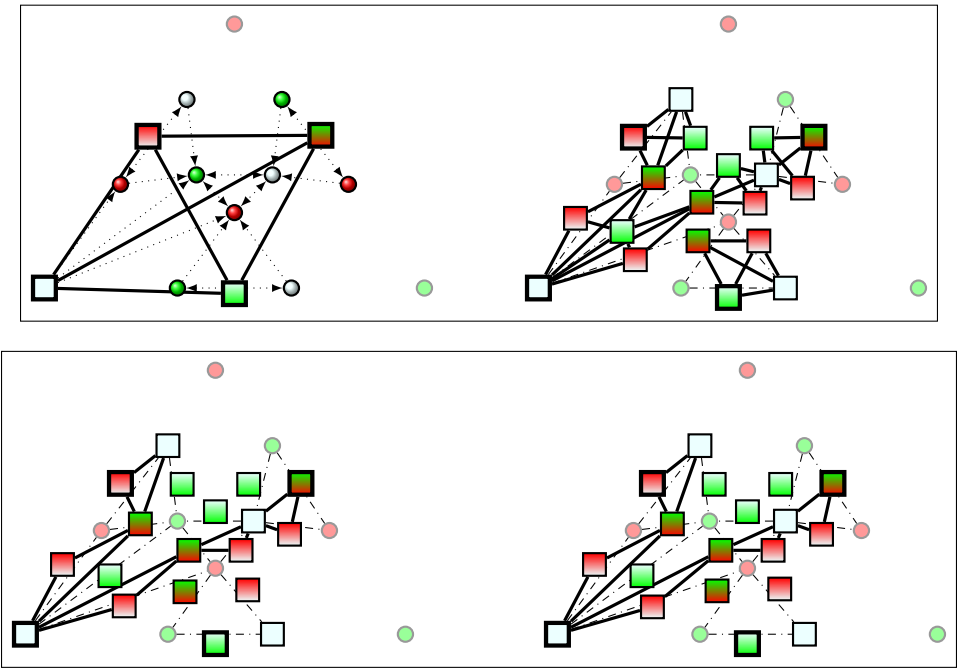

\ctikzfig{Figures/iras1reachability}
\ctikzfig{Figures/iras1reachability2}
\caption{Construction of the root components reachability graphs $\RRG_0$--$\RRG_3$ from the
1-layer protocol complex for the iRAS message adversary: Initial root reachability graph $\RRG_0$ (top left) and
$\RRG_1$ (top right). Since $\RRG_1$ partitions into several connected
components containing incompatible border root components, one has to construct $\RRG_2$ (bottom left). 
As incompatible border root compenents are still connected $\RRG_2$, another
iteration finally provides $\RRG_3=\RRG_2$, so consensus is not solvable here.}
\label{fig:iRASreachability}
\end{figure}

\begin{figure}
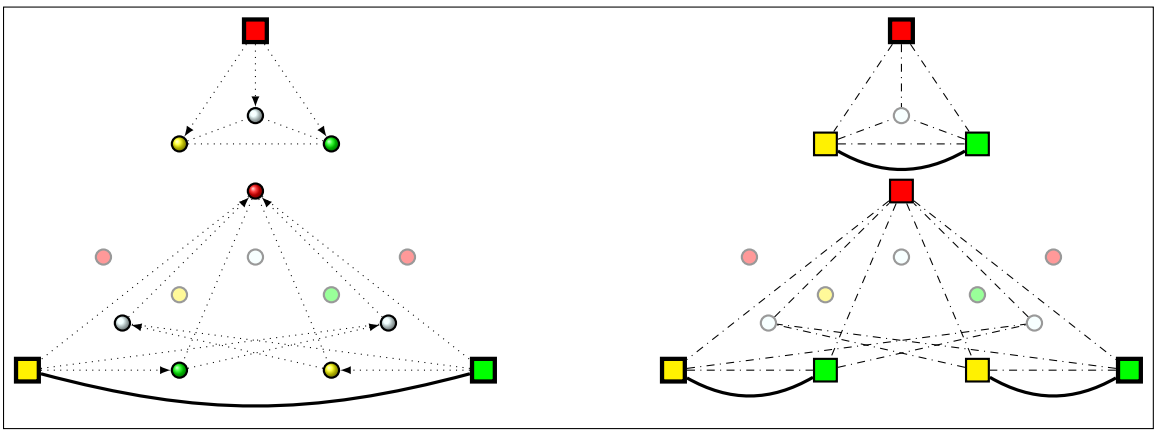

\ctikzfig{Figures/twochainexamplereachability}
\caption{Construction of the root components reachability graphs $\RRG_0$ and $\RRG_1$ from the
1-layer protocol complex for the 2-chain message adversary 2C. 
Since $\RRG_1$ already partitions into connected
components containing only compatible border root components, 
consensus is solvable here. Note that one additional iteration even removes all edges in $\RRG_2$.}
\label{fig:twochainexamplereachability}
\end{figure}

Returning to the example of \cref{fig:multipath}, it is apparent that the root 
reachability graph based decision would not terminate as early as the
nerve graph based procedure, since it explicitly keeps track of all paths between 
border root components. More specifically,
whereas the path $P_a$ between the root components $R_1=\{y\}$ and $R_2=\{w\}$
in $\P_1$ in the top part of \cref{fig:multipath} has vanished in $\P_2$, and
hence also in $\RRG_1$, the path $P_b$ connecting the border root components 
$\{g\}$ and $\{p,w\}$ is lifted to $\P_2$ and hence still present in $\RRG_1$.
Consequently, the decision procedure would proceed to $\RRG_2$ before it can decide that
consensus is solvable. In general, it would faithfully track paths/connected components
that bypass each other until they have been exhausted.

\medskip

It follows that the $\RRG$-based decision procedure would be a natural candidate
for developing a decision procedure that can tell whether distributed consensus
is solvable within $k$ rounds. 
However, like the nerve graph based procedure, it
does not cover delayed path breaking due to case (2). Whereas
a simple way to also accommodate
this would be to scale the number of rounds required for termination by a factor
of $n-1$, i.e., to infer from a number of iterations $k$ of the decision
procedure a consensus termination time bound of $k(n-1)$, this is quite
conservative. The major disadvantage of the $\RRG$-based decision
procedure is its computational
complexity, however: After all, the number of different border root components is 
exponential in $n$ and thus makes the initial root reachability 
graph $\RRG_1$ exponentially larger than the initial nerve graph $\N_0$.

\section{Conclusion}
\label{sec:conclusion}

We presented a topological view on deciding consensus solvability in dynamic graphs controlled by oblivious message adversaries. Compared to the purely combinatorial approach~\cite{WPRSS23:ITCS}, it not only provides additional insights
into the roots of the possible exponential blowup of both the iteration
complexity of the decision procedure and the termination time of distributed
consensus, but also results in a decision procedure for consensus termination
within $k$ rounds. Thanks to our novel concept of a communication pseudosphere,
which can be viewed as the message passing analogon of the chromatic
subdivision, it is also a promising basis for further generalizations,
e.g., for other decision problems and other message adversaries.

\bibliographystyle{splncs04}
\bibliography{lit}
\end{document}

%% file: complexes.tikzstyles
\tikzstyle{VB}=[fill={rgb,255: red,236; green,255; blue,255}, draw=black, shape=circle, thick, minimum size=6mm, inner sep=0pt, font={\footnotesize}, shading=ball, ball color={{rgb,255: red,236; green,255; blue,255}}]
\tikzstyle{VG}=[fill=green, draw=black, shape=circle, thick, minimum size=6mm, inner sep=0pt, font={\footnotesize}, ball color=green]
\tikzstyle{VR}=[fill=red, draw=black, shape=circle, thick, minimum size=6mm, inner sep=0pt, font={\footnotesize}, ball color=red]
\tikzstyle{VY}=[fill=yellow, draw=black, shape=circle, thick, minimum size=6mm, inner sep=0pt, font={\footnotesize}, ball color=yellow]
\tikzstyle{VP}=[fill=pink, draw=black, shape=circle, thick, minimum size=6mm, inner sep=0pt, font={\footnotesize}, ball color=pink]

\tikzstyle{VBsmall}=[fill={rgb,255: red,236; green,255; blue,255}, draw=black, shape=circle, thick, minimum size=2mm, inner sep=0pt, font={\footnotesize}, shading=ball, ball color={{rgb,255: red,236; green,255; blue,255}}]
\tikzstyle{VGsmall}=[fill=green, draw=black, shape=circle, thick, minimum size=2mm, inner sep=0pt, font={\footnotesize}, ball color=green]
\tikzstyle{VRsmall}=[fill=red, draw=black, shape=circle, thick, minimum size=2mm, inner sep=0pt, font={\footnotesize}, ball color=red]
\tikzstyle{VYsmall}=[fill=yellow, draw=black, shape=circle, thick, minimum size=2mm, inner sep=0pt, font={\footnotesize}, ball color=yellow]
\tikzstyle{VPsmall}=[fill=pink, draw=black, shape=circle, thick, minimum size=2mm, inner sep=0pt, font={\footnotesize}, ball color=pink]

\tikzstyle{VBfaded}=[fill={rgb,255: red,236; green,255; blue,255}, draw=black, shape=circle, thick, minimum size=2mm, inner sep=0pt, font={\footnotesize}, opacity=0.4]
\tikzstyle{VGfaded}=[fill=green, draw=black, shape=circle, thick, minimum size=2mm, inner sep=0pt, font={\footnotesize}, opacity=0.4]
\tikzstyle{VRfaded}=[fill=red, draw=black, shape=circle, thick, minimum size=2mm, inner sep=0pt, font={\footnotesize}, opacity=0.4]
\tikzstyle{VYfaded}=[fill=yellow, draw=black, shape=circle, thick, minimum size=2mm, inner sep=0pt, font={\footnotesize}, opacity=0.4]
\tikzstyle{VPfaded}=[fill=pink, draw=black, shape=circle, thick, minimum size=2mm, inner sep=0pt, font={\footnotesize}, opacity=0.4]

\tikzstyle{VRBGsquare}=[top color={rgb,255: red,236; green,255; blue,255},  bottom color=green, middle color=red,draw=black, shape=rectangle, thick, minimum size=3mm, inner sep=0pt, font={\footnotesize}]

\tikzstyle{VBGsquare}=[top color={rgb,255: red,236; green,255; blue,255}, bottom color=green, draw=black, shape=rectangle, thick, minimum size=3mm, inner sep=0pt, font={\footnotesize}]
\tikzstyle{VGRsquare}=[top color=green, bottom color=red,draw=black, shape=rectangle, thick, minimum size=3mm, inner sep=0pt, font={\footnotesize}]
\tikzstyle{VRBsquare}=[top color=red, bottom color={rgb,255: red,236; green,255; blue,255}, draw=black, shape=rectangle, thick, minimum size=3mm, inner sep=0pt, font={\footnotesize}]
\tikzstyle{VBYsquare}=[top color={rgb,255: red,236; green,255; blue,255}, bottom color=yellow, draw=black, shape=rectangle, thick, minimum size=3mm, inner sep=0pt, font={\footnotesize}]
\tikzstyle{VGYsquare}=[top color=green, bottom color=yellow,draw=black, shape=rectangle, thick, minimum size=3mm, inner sep=0pt, font={\footnotesize}]
\tikzstyle{VRYsquare}=[top color=red, bottom color=yellow, draw=black, shape=rectangle, thick, minimum size=3mm, inner sep=0pt, font={\footnotesize}]

\tikzstyle{VBsquare}=[fill={rgb,255: red,236; green,255; blue,255}, draw=black, shape=rectangle, thick, minimum size=3mm, inner sep=0pt, font={\footnotesize}]
\tikzstyle{VGsquare}=[fill=green, draw=black, shape=rectangle, thick, minimum size=3mm, inner sep=0pt, font={\footnotesize}]
\tikzstyle{VRsquare}=[fill=red, draw=black, shape=rectangle, thick, minimum size=3mm, inner sep=0pt, font={\footnotesize}]
\tikzstyle{VYsquare}=[fill=yellow, draw=black, shape=rectangle, thick, minimum size=3mm, inner sep=0pt, font={\footnotesize}]

\tikzstyle{VBGfsquare}=[top color={rgb,255: red,236; green,255; blue,255}, bottom color=green, draw=black, shape=rectangle, ultra thick, minimum size=3mm, inner sep=0pt, font={\footnotesize}]
\tikzstyle{VGRfsquare}=[top color=green, bottom color=red,draw=black, shape=rectangle, ultra thick, minimum size=3mm, inner sep=0pt, font={\footnotesize}]
\tikzstyle{VRBfsquare}=[top color=red, bottom color={rgb,255: red,236; green,255; blue,255}, draw=black, shape=rectangle, ultra thick, minimum size=3mm, inner sep=0pt, font={\footnotesize}]
\tikzstyle{VBYfsquare}=[top color={rgb,255: red,236; green,255; blue,255}, bottom color=yellow, draw=black, shape=rectangle, ultra thick, minimum size=3mm, inner sep=0pt, font={\footnotesize}]
\tikzstyle{VGYfsquare}=[top color=green, bottom color=yellow,draw=black, shape=rectangle, ultra thick, minimum size=3mm, inner sep=0pt, font={\footnotesize}]
\tikzstyle{VRYfsquare}=[top color=red, bottom color=yellow, draw=black, shape=rectangle, ultra thick, minimum size=3mm, inner sep=0pt, font={\footnotesize}]

\tikzstyle{VBfsquare}=[fill={rgb,255: red,236; green,255; blue,255}, draw=black, shape=rectangle, ultra thick, minimum size=3mm, inner sep=0pt, font={\footnotesize}]
\tikzstyle{VGfsquare}=[fill=green, draw=black, shape=rectangle, ultra thick, minimum size=3mm, inner sep=0pt, font={\footnotesize}]
\tikzstyle{VRfsquare}=[fill=red, draw=black, shape=rectangle, ultra thick, minimum size=3mm, inner sep=0pt, font={\footnotesize}]
\tikzstyle{VYfsquare}=[fill=yellow, draw=black, shape=rectangle, ultra thick, minimum size=3mm, inner sep=0pt, font={\footnotesize}]

\tikzstyle{VRBGssquare}=[top color={rgb,255: red,236; green,255; blue,255},  bottom color=green, middle color=red,draw=black, shape=rectangle, thick, minimum size=2mm, inner sep=0pt, font={\footnotesize}]

\tikzstyle{VBGssquare}=[top color={rgb,255: red,236; green,255; blue,255}, bottom color=green, draw=black, shape=rectangle, thick, minimum size=2mm, inner sep=0pt, font={\footnotesize}]
\tikzstyle{VGRssquare}=[top color=green, bottom color=red,draw=black, shape=rectangle, thick, minimum size=2mm, inner sep=0pt, font={\footnotesize}]
\tikzstyle{VRBssquare}=[top color=red, bottom color={rgb,255: red,236; green,255; blue,255}, draw=black, shape=rectangle, thick, minimum size=2mm, inner sep=0pt, font={\footnotesize}]
\tikzstyle{VBYssquare}=[top color={rgb,255: red,236; green,255; blue,255}, bottom color=yellow, draw=black, shape=rectangle, thick, minimum size=2mm, inner sep=0pt, font={\footnotesize}]
\tikzstyle{VGYssquare}=[top color=green, bottom color=yellow,draw=black, shape=rectangle, thick, minimum size=2mm, inner sep=0pt, font={\footnotesize}]
\tikzstyle{VRYssquare}=[top color=red, bottom color=yellow, draw=black, shape=rectangle, thick, minimum size=2mm, inner sep=0pt, font={\footnotesize}]

\tikzstyle{VBssquare}=[fill={rgb,255: red,236; green,255; blue,255}, draw=black, shape=rectangle, thick, minimum size=2mm, inner sep=0pt, font={\footnotesize}]
\tikzstyle{VGssquare}=[fill=green, draw=black, shape=rectangle, thick, minimum size=2mm, inner sep=0pt, font={\footnotesize}]
\tikzstyle{VRssquare}=[fill=red, draw=black, shape=rectangle, thick, minimum size=2mm, inner sep=0pt, font={\footnotesize}]
\tikzstyle{VYssquare}=[fill=yellow, draw=black, shape=rectangle, thick, minimum size=2mm, inner sep=0pt, font={\footnotesize}]

\tikzstyle{Eout}=[->, >=latex]
\tikzstyle{Ein}=[<-, >=latex]
\tikzstyle{Enone}=[-]
\tikzstyle{Einout}=[<->, >=latex]
\tikzstyle{Eoutthick}=[->, >=latex, very thick]
\tikzstyle{Einthick}=[<-, >=latex, very thick]
\tikzstyle{Enonethick}=[-, very thick]
\tikzstyle{Einoutthick}=[<->, >=latex, very thick]
\tikzstyle{Enonedotted}=[-, dash dot]
\tikzstyle{Eoutd}=[->, >=latex, dotted, thin]
\tikzstyle{Eind}=[<-, >=latex, dotted, thin]
\tikzstyle{Enoned}=[-, dotted, thin]
\tikzstyle{Einoutd}=[<->, >=latex, dotted, thin]

%% file: Figures/n2all.tikz
\begin{tikzpicture}
	\begin{pgfonlayer}{nodelayer}
		\node [style=VB] (7) at (-5, -2) {};
		\node [style=VR] (13) at (4, -2) {};
		\node [style=VR] (14) at (-13, -2) {};
		\node [style=VB] (15) at (12, -2) {};
		\node [style=none] (16) at (-13, -3) {\footnotesize $(p_r,(p_r,\{p_r\}))$};
		\node [style=none] (17) at (12, -3) {\footnotesize $(p_w,\{(p_w,\{p_w\})\})$};
		\node [style=none] (18) at (4, -3) {\footnotesize $(p_r,\{(p_r,\{p_r\}),(p_w,\{p_w\})\})$};
		\node [style=none] (19) at (-5, -3) {\footnotesize $(p_w,\{(p_w,\{p_w\}),(p_r,\{p_r\})\})$};
	\end{pgfonlayer}
	\begin{pgfonlayer}{edgelayer}
		\draw [style=Einout] (7) to (13);
		\draw [style=Eout] (14) to (7);
		\draw [style=Enone, bend left=15] (14) to (15);
		\draw [style=Ein] (13) to (15);
	\end{pgfonlayer}
\end{tikzpicture}

%% file: Figures/graphs.tikz
\begin{tikzpicture}
	\begin{pgfonlayer}{nodelayer}
		\node [style=VB] (0) at (-6, 0) {};
		\node [style=VG] (1) at (-2, 2) {};
		\node [style=VR] (2) at (-2, -2) {};
		\node [style=VB] (3) at (3, 0) {};
		\node [style=VG] (4) at (7, 2) {};
		\node [style=VR] (5) at (7, -2) {};
		\node [style=none] (6) at (-3, 0) {};
		\node [style=none] (7) at (6, 0) {};
		\node [style=none] (8) at (-3, 0) {Graph $G_1$};
		\node [style=none] (9) at (6, 0) {Graph $G_2$};
	\end{pgfonlayer}
	\begin{pgfonlayer}{edgelayer}
		\draw [style=Eout] (1) to (0);
		\draw [style=Eout] (2) to (0);
		\draw [style=Eout] (3) to (4);
		\draw [style=Eout] (3) to (5);
	\end{pgfonlayer}
\end{tikzpicture}

%% file: Figures/P0.tikz
\begin{tikzpicture}
	\begin{pgfonlayer}{nodelayer}
		\node [style=VB] (0) at (-6, 0) {};
		\node [style=VG] (1) at (-2, 2) {};
		\node [style=VR] (2) at (-2, -2) {};
		\node [style=none] (3) at (-8.25, 0) {$(p_w,\{p_w\})$};
		\node [style=none] (4) at (0.25, 2) {$(p_g,\{p_g\})$};
		\node [style=none] (5) at (0.25, -2) {$(p_r,\{p_r\})$};
	\end{pgfonlayer}
	\begin{pgfonlayer}{edgelayer}
		\draw (0) to (1);
		\draw (1) to (2);
		\draw (0) to (2);
	\end{pgfonlayer}
\end{tikzpicture}